\documentclass[letterpaper, 10 pt, conference]{ieeeconf}  %

\IEEEoverridecommandlockouts                              %

\usepackage{hyperref}
\usepackage[T1]{fontenc}
\usepackage{epsfig}

\title{\LARGE \bf
Privacy-preserving Nash Equilibrium Synthesis with Partially Ordered Temporal Objectives
}

\author{Caleb Probine$^{1}$ and Abhishek Kulkarni$^{2}$ and Ufuk Topcu$^{1}$%
\thanks{$^{1}$C. Probine and U. Topcu are with the University of Texas at Austin.
}
\thanks{$^{2}$A. Kulkarni is with vijil (was with the University of Texas at Austin). %
        }%
}

\usepackage{amsmath}
\usepackage{amssymb}
\usepackage{xcolor}
\usepackage{comment}
\usepackage{xargs}
\usepackage{acronym}
\usepackage{mathtools}
\usepackage{amsfonts}
\usepackage{paralist}

\usepackage{amsthm}
\acrodef{mdp}[MDP]{Markov Decision Process}
\acrodef{pomdp}[POMDP]{Partially Observable Markov Decision Process}
\acrodef{momdp}[MOMDP]{Multi-objective MDP}
\acrodef{ltl}[LTL]{Linear TtlemporTLal LoTeminating Labeled gic}
\acrodef{dfa}[DFA]{Deterministic finite automaton}
\acrodef{tlmdp}[TLMDP]{terminating labeled Markov decision process}
\acrodef{lmdp}[LMDP]{labeled Markov decision process}
\acrodef{pdfa}[PDFA]{preference deterministic finite automaton}
\acrodef{pdra}[PDRA]{preference deterministic Rabin automaton}
\acrodef{cpltlf}[CPLTL$_f$]{Conditional Preference over LTL$_f$}
\acrodef{cpa}[CPA]{Conditional Preference Automaton}
\acrodef{ltl}[LTL]{linear temporal logic}
\acrodef{ltlf}[LTL$_f$]{linear temporal logic over finite traces}
\acrodef{ne}[NE]{Nash equilibrium}
\acrodef{nea}[NE]{Nash equilibria}

\newcommandx{\ak}[2][1=inline]{\todo[linecolor=blue,backgroundcolor=blue!25,bordercolor=blue,#1]{\scriptsize{[AK]~ #2}}}
\newcommandx{\cp}[2][1=inline]{\todo[linecolor=green,backgroundcolor=green!25,bordercolor=green,#1]{\scriptsize{[CP]~ #2}}}

\newcommand{\ie}{i.e.}

\newcommand{\calP}{\mathcal{P}}
\newcommand{\calM}{\mathcal{M}}

\newcommand{\calD}{\mathcal{D}}
\newcommand{\calU}{\mathcal{U}}
\newcommand{\calE}{\mathcal{E}}

\newcommand{\calY}{\mathcal{Y}}
\newcommand{\calT}{\mathcal{T}}
\newcommand{\calH}{\mathcal{H}}
\newcommand{\calC}{\mathcal{C}}
\newcommand{\calF}{\mathcal{F}}
\newcommand{\calQ}{\mathcal{Q}}

\newcommand{\frU}{\mathfrak{U}}

\newcommand{\bfy}{\mathbf{y}}
\newcommand{\bfu}{\mathbf{u}}

\newcommand{\Paths}{\mathsf{Path}}

\newcommand{\trace}{L}

\newcommand{\Lift}{\mathsf{Semi}}

\newcommand{\PathsInd}[4]{\mathsf{Path}_{#1} (#2,#3,#4)}
\newcommand{\gterm}[1]{{#1}_{\tau}}
\newcommand{\PathSet}[1]{\mathsf{Paths}_{#1}}
\newcommand{\StratSpace}[2]{\Pi_{#1}^{#2}}

\newcommand{\bgobj}[2]{F_{#1}({#2})}
\newcommand{\bgobjP}[3]{F_{#1}^{#3}({#2})}
\newcommand{\lsr}[2]{U_{#1}({#2})}
\newcommand{\lsrP}[3]{U_{#1}^{#3}({#2})}

\newcommand{\qmem}[1]{\calU_{#1}}
\newcommand{\dmempl}[1]{\Gamma_{#1}^+}
\newcommand{\dmemmi}[1]{\Gamma_{#1}^-}
\newcommand{\npvar}{NewPref}

\newcommand{\mpre}{E}
\newcommand{\mPreI}[1]{E_{#1}}
\newcommand{\wpre}{E_{*}}
\newcommand{\wPreI}[1]{E_{#1,*}}
\newcommand{\ppref}[1]{\calE_{#1}}
\newcommand{\upc}[2]{\{#1\}_{#2}^\uparrow}

\newcommand{\qc}{q}
\newcommand{\rc}{m}
\newcommand{\dppl}[1]{\dpc_{#1}^{+}}
\newcommand{\repResp}{\mathsf{R}^*}
\newcommand{\resp}{\mathsf{Resp}}

\newcommand{\dpmi}[1]{\dpc_{#1}^{-}}

\newcommand{\dpc}{\Xi}

\newcommand{\genk}{GenKPrefs}
\newcommand{\genfc}{GenFromConstraints}

\newtheorem{lemma}{Lemma}
\newtheorem{theorem}{Theorem}

\newtheorem{remark}{Remark}

\theoremstyle{definition}
\newtheorem{assumption}{Assumption}
\newtheorem{definition}{Definition}
\newtheorem{proposition}{Proposition}
\newtheorem{problem}{Problem}

\newtheorem{example}{Example}

\newenvironment{proof_alt}[1][Proof]{%
  \par\noindent\textit{#1:} }{\hfill$\blacksquare$\par}

\usepackage{algorithm}
\usepackage[noend]{algorithmic}

\begin{document}

\maketitle
\thispagestyle{empty}
\pagestyle{empty}

\setcounter{footnote}{2}

\begin{abstract}

Nash equilibrium is a central solution concept for reasoning about self-interested agents. 
We address the problem of synthesizing Nash equilibria in two-player deterministic games on graphs, where players have private, partially-ordered preferences over temporal goals. 
Unlike prior work, which assumes preferences are common knowledge, we develop a communication protocol for equilibrium synthesis in settings where players' preferences are private information. 
In the protocol, players communicate to synthesize equilibria by exchanging information about when they can force desirable outcomes. 
We incorporate privacy by ensuring the protocol stops before enough information is revealed to expose a player's preferences.
We prove completeness by showing that, when no player halts communication, the protocol either returns an equilibrium or certifies that none exists. 
We then prove privacy by showing that, with stopping, the messages a player sends are always consistent with multiple possible preferences and thus do not reveal some given secret regarding a player's true preference ordering.
Experiments demonstrate that we can synthesize non-trivial equilibria while preserving privacy of preferences, highlighting the protocol’s potential for applications in strategy synthesis with constrained information sharing.
\footnote{\label{fn:extended_version}We give extended details for proofs and experiments in the appendix.}

\end{abstract}

\section{Introduction}

Sequential interactions among self-interested agents arise across domains such as planning, negotiation, and control. 
These interactions demand strategies that not only pursue individual goals but also adapt to the dynamic choices of others. 
Nash equilibrium serves as a key solution concept in these settings. 
At Nash equilibrium, no player benefits from a unilateral deviation from their strategy. 

We develop algorithms for privacy-preserving Nash equilibrium synthesis in two-player games on graphs, where players hold partial preferences on temporal objectives. 
Unlike more commonly studied total orders, partial preferences allow agents to express incomparable outcomes \cite{rahmani_probabilistic_2023}, an important feature in complex environments.
By encoding preferences over the states of finite automata, we capture rich specifications of temporal goals, such as those describable in finite-trace linear temporal logic (LTL$_f$) \cite{giacomo_linear_nodate}. 

Preserving privacy is essential in domains such as healthcare, autonomous logistics, and competitive markets. Agents may be unwilling to share preferences derived from sensitive data, or may wish to withhold strategic information for future interactions \cite{kearns_private_2012}. These reasons motivate the need to design Nash equilibrium synthesis methods that reduce the information revealed about players' preferences. 

Previous work, such as \cite{kulkarni_nash_2024}, provides methods for synthesizing Nash equilibria in games where players have preferences on temporal goals.
However, these approaches assume the players' preferences are common knowledge.

We design a private communication protocol that enables two agents to compute a Nash equilibrium by exchanging partial preference information. 
At each round, an agent communicates a solution to a zero-sum game derived from its preferences, which reveals whether the agent can guarantee an outcome it prefers strictly over some reference outcome. 
Importantly, the protocol includes a privacy check where each agent monitors whether continued communication risks revealing sensitive aspects of its preference structure. 

We make the following contributions.
\begin{itemize}
    \item We develop a protocol for Nash equilibrium synthesis in games on graphs where players have partial preferences on temporal objectives, and we incorporate a stopping mechanism  to avoid privacy violations.
    \item We prove the protocol's completeness by showing that, when no player stops the protocol, it either computes a Nash equilibrium or certifies that none exists.
    \item We prove the protocol's privacy by showing that, with the stopping mechanism, a player's messages in the protocol do not reveal a given secret.
    \item Through a package-delivery example, we show that we can synthesize non-trivial Nash equilibria without violating privacy constraints.
\end{itemize}
The methods we present lay foundations for further work on privacy-preserving synthesis in games with preferences.

\subsection*{Related Work}

\textbf{Preferences and games. }
Various works explore normal-form games where players have partial preferences, e.g.,  \cite{rozen_games_2018,bade_nash_2005,zanardi_posetal_2022}.
We study games on graphs, where players have partial preferences on the satisfaction of temporal objectives.

Existing work provides characterizations and synthesis methods for Nash equilibria in games on graphs where players have preferences on temporal objectives.
General results on \acf{ne} existence include \cite{le2013infinite,le2014winning,bruyere2017existence,le2018extending}, but these results are either non-constructive or make strict assumptions on preferences. 
\ac{ne} synthesis methods exist for LTL, B\"uchi, and reachability objectives \cite{fu2015concurrent,bouyer_pure_2015,guelev2025qualitative} and for settings with preferences defined on infinite words via automata \cite{bruyere2025games}.
We study games where players have preferences on finite words,  %
specified via semi-automata, and this setting includes \ac{ne} synthesis when players have preferences on
LTL$_f$ objectives.
When players have opposite preferences on LTL$_f$ objectives, one can compute a \ac{ne} by computing non-dominated-almost-sure winning strategies \cite{noauthor_sequential_2024}. %
For general preferences, one can compute a \ac{ne} by considering the alignment of the players' preferences \cite{kulkarni_nash_2024}. 
However, these methods assume access to both players' preferences and do not ensure players' privacy.

\textbf{Privacy and equilibrium computation. }
Prior work explores how one can compute equilibria while maintaining agents' privacy.
For example, one can maintain differential privacy when computing correlated equilibria of large normal-form games \cite{kearns_private_2012} or when computing Nash equilibria in congestion games \cite{rogers_asymptotically_2014} and aggregative games \cite{cummings_privacy_2015}.
While the above equilibrium computation approaches %
use a mediator,
various distributed differentially-private \ac{ne} synthesis methods also exist, e.g., \cite{ye_differentially_2022,huo_compression-based_2024,lin_statistical_2024,wang_differentially_2024-1,chen2025differentially}.
None of the above approaches address games on graphs or partial preferences.
Existing work explores the distinguishability of players' objectives in equilibrium play on games on graphs 
\cite{nakanishi2025strategies}.
However, the above approach conceals information revealed by the equilibrium solution itself rather than ensuring privacy during synthesis. Furthermore, the approach does not apply to settings with preferences on temporal objectives.

While not motivated by privacy, distributed iterative schemes for computing a NE, such as best-response dynamics, remove the requirement that preferences are common knowledge \cite{li_distributed_1987}. 
While prior work explores the convergence of such dynamics in games on graphs \cite{andersson_acyclicity_2010}, the results do not guarantee privacy and do not extend to settings where players have partial preferences on temporal objectives.

\section{Preliminaries}

\subsection{Interaction Model}

We model the interaction as a deterministic two-player turn-based game on a graph \cite{fijalkow2023games}, extended with a termination mechanism. 
The termination mechanism ensures that when evaluating satisfaction of temporal goals using a finite trace, as in LTL$_f$ \cite{giacomo_linear_nodate}, we use the same trace for each player.

\begin{definition}
    A deterministic two-player turn-based game $G$ is a tuple
	$
        \langle S
        = (S_1 \cup S_2)
        , s_0
        , A , T, AP, L \rangle.
	$
    $S$ is a state set partitioned into player 1 states $S_1$ and player 2 states $S_2$. 
    The initial state is $s_0$.
    $A = A_1 \cup A_2$ is a similarly partitioned finite set of actions. 
    $T: S \times A \rightarrow S$ is a transition function. %
    $AP$ contains atomic propositions.
	$L: S \rightarrow 2^{AP}$ is a labeling function.
    Given game $G$, we define a deterministic two-player turn-based game  $G_\tau$ \emph{with termination} as a tuple
    $$
		\gterm{G} = 
        \langle (S_1 \cup S_2)
        \times \{0, 1\} 
        , (s_0,0)
        , A \cup \{\tau\}, T', AP, L \rangle.
	$$
    $S\times \{0,1\}$ is an augmented state space where $(s_0,0)$ is the initial state, and $(s, 1)$ is a terminating state for any $s \in S$.
    We add \emph{termination action} $\tau$ to $A$.
    $T'$ is an augmented transition function %
    on $S\times \{0,1\}$
    where $T'((s,0),a) = (T(s,a),0)$ for all $s,a \in S\times A$, $T((s, j), \tau) = (s, 1)$ for $(s,j) \in S\times \{0,1\}$, and $(s,1)$ is absorbing for $s\in S$. 
        
\end{definition}

\textbf{Paths.} 
$\gterm{G}$ models game play until a player selects the action $\tau$, 
and the game moves to the absorbing state $(s,1)$.
A \emph{path} $\rho = (s_0,j_0)(s_1,j_1)\ldots$ is a 
sequence where, for each $i \geq 0$, $(s_{i+1}, j_{i+1}) = T'((s_i, j_i), a_i)$ for some $a_i \in A$. 
A path is \emph{terminating} if $j_i = 1$ for some $i \geq 0$.
The terminating state is $\mathsf{Last}(\rho)= (s_k, 1)$, where $k \geq 1$ is the smallest integer such that $j_k = 1$.
The trace $L(\rho)$ of a terminating path $\rho$ is a finite word in $\Sigma = {2^{AP}}^*$, comprised of the labels before the terminating state, \ie \ $\trace(\rho) = L(s_0) L(s_1) \cdots L(s_{k-1})$. 
We define $\mathsf{Occ}(\rho)$ as the set of states a path visits.

\textbf{Strategies.}
A strategy $\pi_i: (S \times \{0\})^* (S_i \times \{0\}) \rightarrow A_i \cup {\tau}$, for $i$ in $\{1,2\}$, selects actions for player $i$. 
A \emph{memoryless} strategy $\pi_i$ depends only on the current state, \ie, $\pi_i: (S_i \times \{0\}) \rightarrow A_i \cup {\tau}$. Given a starting state $(s,0)$, a strategy pair $(\pi_1, \pi_2)$ determines a unique path $\mathsf{Path}_{G_\tau}((s,0),\pi_1,\pi_2)$.
If omitted, we assume the state is the initial state, \ie, 
$\mathsf{Path}_{G_\tau}(\pi_1,\pi_2) = \mathsf{Path}_{G_\tau}((s_0,0),\pi_1,\pi_2)$
A path $\rho$ is \emph{consistent} with strategy $\pi_1$ if there exists $\pi_2'$ such that $\mathsf{Path}_{G_\tau}((s_0,0),\pi_1,\pi_2') = \rho$.
We introduce proper strategies \cite{patek1999stochastic} to reason about termination.
With a proper strategy, players ensure that either, they eventually terminate the game, or the game eventually leaves their states.

\begin{definition}
	A strategy $\pi_i$ is \emph{proper} if,
    for any path $\rho$ consistent with $\pi_i$, we have the following.
	\begin{itemize}
		\item $\exists n \geq 0$ such that $\pi_i(\rho[0] \ldots \rho[n]) = \tau$, or
		\item $\exists N$ such that for all $n \geq N$, $\rho[n]$ is in $S \setminus S_{i} \times \{0, 1\}$.
	\end{itemize}
\end{definition}

If all players use proper strategies, the game will terminate.

\begin{proposition}
    \label{prop:ProperPolicyTerm}
    If strategies $\pi_1$ and $\pi_2$ are proper, then $\mathsf{Path}_{G_\tau}((s,0),\pi_1,\pi_2)$ terminates.
\end{proposition}

\paragraph*{Proof sketch\footref{fn:extended_version}}
We assume $\rho$ never visits $S \times \{1\}$, and demonstrate a contradiction.

\textbf{B\"uchi games.}
Given game arena $\langle S=S_1\cup S_2,A,T \rangle$ comprising states, actions, and transitions, a B\"uchi objective is a subset $F \subset S$. A winning strategy for B\"uchi objective $F$ ensures states in $F$ appear infinitely often.

\subsection{Preference Modeling}

We model preferences on temporal goals with \emph{preference automata} \cite{rahmani_probabilistic_2023}.
A preorder encodes preferences on a set.
\begin{definition}
    \label{def:preorder}
    A preorder on a set $U$ is a binary relation on $U$ that is reflexive and transitive.
\end{definition}
\noindent For a preorder $\mpre$ on $Q$, $q \succeq_\mpre q'$ if and only if $(q,q') \in E$.
For $q,q' \in Q$, $q \succ_\mpre q'$ denotes that $q$ is strictly preferred to $q'$, \ie \  $q \succeq_\mpre q'$ and $q' \not\succeq_\mpre q$.
The set of preorders on $Q$ is $\calP^Q$.
For preorder $\mpre$ 
and element $q \in Q$, the set $\{q\}^{\uparrow}_{\mpre} = \{q' \in Q: q' \succ_{\mpre} q\}$ is a strict upper closure of $E$.
$\mathcal{E}_\emptyset^Q = \{(q,q)| q \in Q\}$ is the smallest preorder on $Q$ while $\smash{\mathcal{E}_{(x,y)}^Q} = \mathcal{E}_\emptyset^Q \cup \{(x,y)\}$ is the smallest preorder such that $x \succ_{\calE} y$.

\begin{definition}
	\label{def:preference-automaton}
	A tuple
	$\langle Q, \Sigma, \delta, q_0, E \rangle$ is a \emph{preference automaton}, where
    $Q$ is a finite set of states,
    $\Sigma = 2^{AP}$ is the alphabet,
    $\delta: Q \times \Sigma \rightarrow Q$ is a deterministic transition function,
	$q_0 \in Q$ is an initial state, and
    $E \subseteq Q \times Q$ is a preorder on $Q$.
\end{definition}

The \emph{semi-automaton} $\langle Q, \Sigma, \delta, q_0 \rangle$ tracks progress toward temporal goals, and the preorder $E$ encodes preferences on goals.
One can construct preference automata from preferences on LTL$_f$ goals \cite{rahmani_probabilistic_2023}.
Abusing notation, we use $\delta$ for the extended transition function $\delta: Q\times \Sigma^* \rightarrow Q$ where
$\delta(q,\sigma w) = \delta(\delta(q,\sigma),w)$ for $w \in \Sigma^*$ and $\sigma \in \Sigma$, and $\delta(q,\epsilon) = q_0$, for $\epsilon$ empty.
The preorder $\wpre$ on $\Sigma^*$ satisfies $(w,w') \in \wpre$ if and only if $(\delta(q_0,w),\delta(q_0,w')) \in E$.

The following example illustrates how a preference automaton encodes preferences over temporal goals.

\begin{example}
	\begin{figure}[h]
		\centering
		\includegraphics[scale=1]{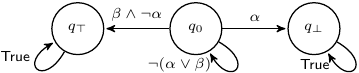}
		\caption{A semi-automaton over atomic propositions $\{\alpha,\beta\}$ that encodes if some event $\alpha$ occurs before some $\beta$.}
		\label{fig:pdfa_example}
	\end{figure}

    Figure~\ref{fig:pdfa_example} shows a semi-automaton to encode goals relating to the order of two events $\alpha$ and $\beta$.

    The preference automaton encodes preferences on goals
    through the preorder. 
    For example, we can specify that a player's most preferred outcome is to ensure $\alpha$ occurs strictly before $\beta$, by 
    specifying preorder $\mpre$ such that $q_\bot \succ_{\mpre} q_\top$.
\end{example}

\section{Problem Formulation}
\label{sec:problem}

For games where players have preferences on temporal goals, encoded via preference automata, we aim to find Nash equilibria while preserving privacy of players' preferences.

\textbf{Nash equilibrium.} 
We define Nash equilibria via preferences on paths.
Given proper strategies $\pi_1$ and $\pi_2$ 
, the unique terminating path $\rho = \Paths_{G_\tau}((s_0,0), \pi_1, \pi_2)$ determines the game's outcome.
Players evaluate this outcome using their preferences over traces, \ie, the sequences of labels induced by terminating paths.
Given words $w, w' \in \Sigma^*$, $w \succ_{\wPreI{i}} w'$ denotes that word $w$ is strictly preferred to $w'$ under player $i$'s preferences.
We define a Nash equilibrium in $\gterm{G}$ as follows.
\begin{definition}%

	Proper strategy pair $(\pi_1, \pi_2)$ is a Nash equilibrium (NE) if no player $i$ can deviate to another proper strategy $\pi_i'$ and induce a strictly preferred outcome. That is,
	\begin{align*}
		&\nexists \pi_1': L(\Paths_{\gterm{G}}( \pi_1', \pi_2)) \succ_{\wPreI{1}} L(\Paths_{\gterm{G}}(\pi_1, \pi_2)), \\
		&\nexists \pi_2': L(\Paths_{\gterm{G}}(\pi_1, \pi_2')) \succ_{\wPreI{2}}  L(\Paths_{\gterm{G}}( \pi_1, \pi_2)),
	\end{align*}
\end{definition}

\textbf{Privacy.} 
We consider games where players' preferences over the states of the semi-automaton are private knowledge.  

\begin{assumption}
	The semi-automaton $\langle Q, \Sigma, \delta, q_0 \rangle$ 
    is common knowledge, but preorders $E_1$ and $E_2$ are known only to the respective players. 
\end{assumption}

As preferences are private, this setting represents an incomplete information game  \cite{harsanyi1967games}. 
However, we are interested in whether the players can find an NE in a complete information game, where players' preferences are common knowledge,
through privacy-preserving communication.

A \emph{privacy-preserving} communication protocol in $\gterm{G}$ identifies a NE without requiring players to reveal private preferences. 
Players interact with a mediator by sending messages from a set $\calM$ in response to queries from a set $\calY$.
Set $\calM$ includes a $\mathsf{STOP}$ message, so that players may leave the protocol to preserve privacy.
Formally, a mediator is a function $\mathsf{Med}: \calY^* \times \calM^* \rightarrow \calY$ that defines the next query.
Players respond according to $\resp : \calP^Q \times \calY^* \rightarrow \calM$, a function defining the next response, for a given preference. 
Note that $\resp$ will also implicitly depend on the game $G_\tau$ and the player.
$\repResp: \calP^Q \times \calY^* \rightarrow \calM^*$ is the induced mapping from query sequences to response sequences.

For a given secret, privacy holds for a player if they do not reveal the truth-value of the secret by participating in the protocol.
Formally, a secret $\frU$ is a subset of the set of preorders $\calP^Q$, and player $i$ wants to ensure that communication does not reveal whether $E_i \in \frU$ or $E_i \in \frU^c$.

\begin{definition}
    A pair $(\mathsf{Med}, \mathsf{Resp})$ defines a privacy-preserving NE synthesis protocol if it satisfies the following.
    \begin{enumerate}
        \item If no player replies $\mathsf{STOP}$, the replies allow the mediator to construct an equilibrium or certify none exist.
        \item For any query sequence $\bfy \in \calY^*$ and preorder $\mpre$, there exist preorders $\tilde{\mpre}^+$ in $\frU$ and $\tilde{\mpre}^- $ in $ \frU^c$ so that $\repResp(\mpre,\bfy) = 
        \repResp(\tilde{\mpre}^+,\bfy) = 
        \repResp(\tilde{\mpre}^-,\bfy)$. That is, engaging in the protocol does not reveal the secret's value.
    \end{enumerate}
\end{definition}
We note that we define the mediator for convenience, but one can assume one of the players to choose the queries.

We consider the problem of designing such a protocol.

\begin{problem}
\label{prob:problem}
	Given a game $\gterm{G}$, 
    semi-automaton $\langle Q, \Sigma, \delta, q_0 \rangle$,
    and a secret $\frU \subset \calP^Q$, design a privacy-preserving NE synthesis protocol $(\mathsf{Mediator}, \resp)$.
\end{problem}

\section{Nash equilibrium synthesis}
\label{sec:NashSynth}

Following reactive synthesis approaches \cite{baier2008principles}, we define a product game to track the current automaton state.

\begin{definition}
{
Let $(Q,\Sigma, \delta, q_0)$ be a semi-automaton for some preference automaton, and let a game $\gterm{G}$ be given.
The \emph{terminating product game} is the tuple
\begin{equation*}
    \gterm{H} = \langle V',A',\Delta, (v_0,0) \rangle.
\end{equation*}
$V' = V\times \{0,1\} = S\times Q\times \{0,1\}$ is a state set, and we 
define $V_i = S_i \times Q$ for $i$ in $\{1,2\}$. $A' = A\cup \{\tau\}$ is a set of actions. $\Delta: V'\times A' \rightarrow V'$ is a deterministic transition function.
For $a \in A$, $\Delta((s,q,0),a) = (s',q',0)$ if and only if $(s', 0) = T'((s, 0), a)$ and $q' = \delta(q, L(s'))$, where $T'$ is the transition function of $G_\tau$. 
For a state $(s,q,0)$, $\Delta((s,q,0),\tau) = (s,q,1)$, and states in $V\times\{1\}$ are absorbing. 
The semi-automaton portion $q$ of the state does not evolve at or after termination. 
The initial state is $(v_0,0)$ where $v_0 = (s_0,\delta(q_0,L(s_0)))$.
We define $\Lift$ to return the semi-automaton part of a state such that $\Lift(s,q) = q$.
}

\end{definition}

As $(Q,\Sigma,\delta,q_0)$ is common knowledge, both players can construct $\gterm{H}$. 
However, players have distinct preferences on the states of $H_\tau$. 
For $i \in \{1, 2\}$, player $i$ privately computes a preorder $\ppref{i}$
on $V$ using $E_i$ such that
$((s,q),(s',q')) \in \ppref{i}$ if and only if $(\Lift(s,q), \Lift(s',q')) = (q,q') \in E_i$.

A terminating path $\rho = (s_0,0) \ldots (s_k,0) (s_k,1) \ldots$ in $\gterm{G}$ induces a terminating path $\varrho = (v_0,0) \ldots (v_k,0) (v_k,1) \ldots$ in $\gterm{H}$, where $v_j = (s_j,q_j)$, and  $q_j = \delta(q_0,L(s_0) \ldots L(s_j))$ for $j$ in $0,\ldots,k$. 
We refer to $\varrho$ as the trace of $\rho$ in $\gterm{H}$.
For paths in $\gterm{G}$, the terminating states of the corresponding traces in $\gterm{H}$ determine a player's preference on them.

Thus, finding a NE in $\gterm{G}$ reduces to finding a terminating state in $\gterm{H}$ from which no player can unilaterally deviate to terminate the game in a
state they strictly prefer.

\subsection{Equilibrium Existence with Proper Strategies}

We restrict players to proper strategies to ensure that game paths always have interpretations
in the semi-automaton.

We use the game in Figure~\ref{fig:cycExample} to show that restricting players to proper strategies leads to interesting \ac{ne} behavior.

\begin{figure}[h]
\centering
\includegraphics[scale=1]{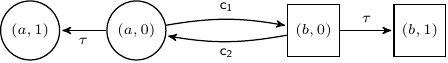}
\caption{A game, adapted from \protect\cite{bouyer2016stochastic}, where proper strategies induce interesting equilibrium properties. 
Player 1's states are circles, player 2's states are squares, and actions are transition annotations. 
The initial node is $(a,0)$.
We set $L(a) = \alpha$ and $L(b) = \beta$. 
The semi-automaton $\langle \{q_1,q_2\},\{\alpha,\beta\},\delta, q_1 \rangle$ tracks the last node, \ie,  $\delta(\cdot, \alpha) = q_1$ and $\delta(\cdot,\beta) = q_2$.
Each player prefers not to be the one terminating the game, \ie, 
$q_2 \succ_{E_1} q_1$ and $q_1 \succ_{E_2} q_2$.
}
\label{fig:cycExample}
\end{figure}

The proper strategy restriction prevents players from using infinite cycles to enforce desired outcomes. 
Suppose player 2 uses a memoryless strategy with $\pi_2((b,0)) = \mathsf{c}_2$. 
In this case, the game cycles forever unless player 1 terminates the game in $(a,1)$, achieving player 2's preferred outcome. 
However, $\pi_2$ is not proper,
If player 2 uses a proper strategy, player $1$ has a proper strategy to ensure the game terminates at $(b,1)$.

When players can not use infinite cycles to enforce desired outcomes, equilibria may not exist.
Indeed, in Figure~\ref{fig:cycExample}, each player can force the game to terminate in a chosen state if the other player uses a proper strategy.
Thus, one player can always deviate to induce their strictly preferred outcome.

\subsection{Synthesizing Nash Equilibria via Winning Region Sharing}
\label{sec:NashSynthMain}

We give an algorithm to synthesize equilibria using access to a query model that returns the solution of a certain zero-sum game.
This algorithm builds on prior work that constructs equilibria from solutions of zero-sum games \cite{bouyer_pure_2015}.

We subdivide the problem by the final semi-automaton state, which represents the outcome of the game. 
For each $q \in Q$, we give a \ac{ne} strategy pair where the game terminates in $S \times \{q\}$, or certify that no such pair exists.

To find a \ac{ne} strategy pair where the game terminates in $S \times \{q\}$, we solve a B\"uchi game. 
For $q\in Q$, we aim to find a state $w \in S \times \{q\}$, a \emph{nominal} path $\bfu$ from $(v_0,0)$ to $(w,1)$, and a pair of strategies that ensure no player has incentive to deviate from $\bfu$.
Finding $\bfu$ reduces to finding a path in the product game that remains in the losing region of a pair of B\"uchi games.
For $i$ in $\{1,2\}$ we first define $V^+_i(q)$ with
\begin{equation*}
    V^+_i(q) = \{v\in V: \Lift(v) \succ_{\mPreI{i}} q\} 
\end{equation*}
and we define $V^-_i(q) = V \setminus V_i^+(q).$

We define a B\"uchi objective on the product game $H_\tau$, such that player $i$ wants to ensure the following set of nodes appears infinitely often.
\begin{equation}
    \bgobj{i}{q} = \{(v,j) \in V_i^+(q) \times \{0,1\}\}.
\end{equation}

We design player $i$'s B\"uchi game so that, once the game is in player $i$'s winning region, they can force the game to terminate in an outcome they prefer to $q$.
Indeed, if $(v,1)$ appears,
for $v \in V^+_i(q)$, then the game has terminated in an outcome player $i$ strictly prefers to $q$.
Meanwhile, if $(v,0)$ appears infinitely often for some $v \in V_i^+(q)$, we can use the fact that the other player is using a proper strategy, to argue that 
player $i$ has a strategy to terminate the game in $V^+_i(q)$.

In this setting, 
B\"uchi 
games provide a way to account for the nuances that proper strategies introduce.

In player $i$'s losing region $\lsr{i}{q}$, \ie \ the winning region's complement,
the other player has a strategy $\sigma_{-i}$ to ensure the game terminates in a state player $i$ does not strictly prefer to $q$. 
As in \cite{brihaye_multiplayer_2012}, we term these strategies \emph{punishment} strategies.

We find \ac{ne} strategies by computing nominal paths that reach $S \times \{q\} \times \{1\}$ 
while remaining in the intersection $\lsr{1}{q} \cap \lsr{2}{q}$ of the players' losing regions.
If a player deviates from the nominal path, the other player uses the punishment strategy to ensure the deviating player does not achieve an outcome they strictly prefer to $q$.
Thus, the nominal path and punishment strategies define a \ac{ne}.

Algorithm~\ref{alg:MainAlg} encodes this approach to computing a \ac{ne}. The algorithm iterates over semi-automaton states $q$ and attempts to find a path that remains in $\lsr{1}{q} \cap \lsr{2}{q}$. 

\begin{algorithm}[t]
\caption{\ac{ne} synthesis via winning regions. 
}
\label{alg:MainAlg}
\begin{algorithmic}[1]
\FOR{$q \in Q$}
    \STATE \textcolor{orange}{\textbf{Query}} player $1$ for losing region $\lsr{1}{q}$ %
    \STATE \textcolor{orange}{\textbf{Query}} player $2$ for losing region $\lsr{2}{q}$ %
    \IF{$\exists \mathbf{u} = (v_0, 0) \ldots (w, 1) $ such that $\ldots$ \\ $ \mathsf{Occ}(\mathbf{u}) \subseteq (\lsr{1}{q} \cap \lsr{2}{q}) \wedge \Lift(w) = q$}
        \STATE \textcolor{orange}{\textbf{Query}} player $1$ for a memoryless punishment strategy $\sigma_{-1}$ in $\mathsf{Buchi}(\bgobj{1}{q})$
        \STATE \textcolor{orange}{\textbf{Query}} player $2$ for a memoryless punishment strategy $\sigma_{-2}$ in $\mathsf{Buchi}(\bgobj{2}{q})$
        \STATE \textbf{return} a Nash equilibrium using $\mathbf{u}$, $\sigma_{-1}$ and $\sigma_{-2}$ \label{line:NashConst}
    \ENDIF
\ENDFOR
\RETURN \FALSE
\end{algorithmic}
\end{algorithm}

When viewed as a protocol, a mediator decides on states to query and constructs a NE from responses. The agents implement response algorithms that return game information while maintaining privacy.
In orange,
we highlight points in Algorithm~\ref{alg:MainAlg} where the mediator queries players.

\begin{remark}
    \label{rem:Restriction}
    As $S \times \{q\} \times \{1\}$ is always in $\lsr{i}{q}$, player $i$ only needs to share the intersection of $\lsr{i}{q}$ with $V_r \cap (V \times \{0\})$ where $V_r$ are the reachable nodes of the product game. 
\end{remark}

\begin{remark}
If $m$ is the number of edges in the graph of $H_\tau$,
Algorithm~\ref{alg:MainAlg} has complexity $\mathcal{O}(m|V|^2)$.
We solve at most $|Q|$ B\"uchi games, which take $\mathcal{O}(m|V|)$ time to solve \cite{fijalkow2023games}.
\end{remark}

\subsection{Completeness of Algorithm~\ref{alg:MainAlg}}

We construct proper \ac{ne} strategies from the path $\mathbf{u}$ and the strategies $\sigma_{-1}$ and $\sigma_{-2}$.
We first define $\hat{\sigma}_{-i}$ as follows.
\begin{equation*}
    \hat{\sigma}_{-i}(v,0) = M_i({\sigma}_{-i})(v,0) = 
    \begin{cases}
       \tau & \text{ if } v \in V_i^-(q) \\ 
       {\sigma}_{-i}(v,0) & \text{ otherwise.}
    \end{cases}
\end{equation*}
The strategy $\hat{\sigma}_{-i}$ is a punishment strategy too, and satisfies a condition that will help construct proper strategies from $\hat{\sigma}_{-i}$.
\begin{lemma}
\label{lem:noRepMain}
    If $\sigma_{-2}$ is a punishment strategy in $\mathsf{Buchi}(\bgobj{2}{q})$, 
    then $\hat{\sigma}_{-2} = M_2(\sigma_{-2})$ is also a punishment strategy.
    Also, for any $u\in \lsr{2}{q}$ and $\pi_2 \in \Pi_2^{\gterm{H}}$, $\rho = 
    \PathsInd{\gterm{H}}{u}{\hat{\sigma}_{-2}}{\pi_2}$ visits each vertex in $V_1\times\{0\}$ at most once.
\end{lemma}

\textit{Proof sketch\footref{fn:extended_version}:} 
Strategy $\hat{\sigma}_{-2}$ differs from $\sigma_{-2}$ on paths that visit $V_2^-(q) \cap V_1$, but if this event occurs, the game terminates in $V_2^-(q)$, player $1$ wins, and player $2$ loses. For the second part,
we use memorylessness of $\hat{\sigma}_{-2}$.
If $\mathsf{Path}_{H_\tau}(u, \hat{\sigma}_{-2}, \pi_2)$ 
visits a node $(v,0)$ in $ V_1 \times \{0\}$ twice for some $\pi_2$, then $\mathsf{Path}_{H_\tau}(u, \hat{\sigma}_{-2}, \pi_2')$ visits $v$ infinitely often for some $\pi_2'$.
As $\hat{\sigma}_{-2}$ is winning for player $1$, $v$ must lie in $V_1 \cap V_2^-$, which is a contradiction, as $\hat{\sigma}_{-2}$ terminates the game at $v$. 
$\square$

\noindent Note that $\hat{\sigma}_{-i}$ depends on $V_i^-(q)$, and so requires extra private information from player $i$ to construct. 
We assume player $i$ processes $\sigma_{-i}$ to get $\hat{\sigma}_{-i}$ before sharing it.

We can combine $\mathbf{u}$ with $\hat{\sigma}_{-i}$ to construct a full proper strategy.
Informally, let $C_1(\mathbf{u},\hat{\sigma}_{-2})$ be the strategy for player $1$ that follows the path $\mathbf{u}$ until player $2$ deviates, at which point player $1$ plays $\hat{\sigma}_{-2}$. This strategy is proper.
\begin{proposition}
    \label{prop:ProperCombination}
    The strategy $C_1(\mathbf{u}, M_2(\sigma_{-2}))$ is proper.
\end{proposition}

\textit{Proof sketch\footref{fn:extended_version}: } 
A non-terminating path deviates from $\mathbf{u}$, at which point, player $1$ uses $\hat{\sigma}_{-2}$. After deviation, the path visits $V_1\times \{0\}$ finitely many times by Lemma~\ref{lem:noRepMain}. 
$\square$

\noindent We define $C_2(\mathbf{u},\hat{\sigma}_{-1})$ analogously and note that versions of Lemma~\ref{lem:noRepMain} and Proposition~\ref{prop:ProperCombination} with flipped players also hold.

Strategies $C(\mathbf{u},\hat{\sigma}_{-2})$ and $C(\mathbf{u},\hat{\sigma}_{-1})$ then comprise a NE.
\begin{theorem}
    \label{thm:NESuff}
    If Algorithm~\ref{alg:MainAlg} terminates, the pair $(\pi_1,\pi_2) = (C_1(\mathbf{u}, M_2(\sigma_{-2})),  C_2(\mathbf{u}, M_1(\sigma_{-1})))$ is a NE. 
\end{theorem}

\textit{Proof sketch\footref{fn:extended_version}: } 
$\PathsInd{\gterm{H}}{(v_0,0)}{\pi_1}{\pi_2}$
trivially reaches $(w,1)$.
If player $2$ deviates with a proper strategy, player $1$ follows the proper punishment strategy $\hat{\sigma}_{-2}$, and the game must terminate in $V_2^-(q)$, as $\hat{\sigma}_{-2}$ is winning for player $1$ in player $2$'s B\"uchi game from any node in $\mathbf{u}$. 
$\square$

Algorithm~\ref{alg:MainAlg} is also necessary in the following sense. 
\begin{theorem}
    \label{thm:NENec}
    Suppose Algorithm~\ref{alg:MainAlg} returns \textbf{false}. Then no proper strategy NE exists in $\gterm{H}$. 
\end{theorem}

\textit{Proof sketch\footref{fn:extended_version}: } 
Any proper strategy pair $(\sigma_{1},\sigma_{2})$ visits one player's winning region, so one player always has a potentially non-proper strategy $\sigma$ to win their B\"uchi game. 
We construct a proper strategy that wins the B\"uchi game using the fact that the other player follows a proper strategy.

\section{Privacy preservation methods}
\label{sec:priv_methods}

Given a secret $\frU$, \ie, a set of preorders on semi-automaton states $Q$, 
player $i$ must ensure the losing regions and punishment strategies they share, based on preorder $E_i$, do not reveal whether $E_i \in \frU$ or not.
Each player uses a protocol $\resp$ that defines a mapping $\mathsf{R}^*: \calP^Q \times \calY^* \rightarrow \calM^*$ from preferences and query sequences 
to response sequences.
Example queries and responses are semi-automaton states and losing regions, respectively.
A private protocol satisfies
\begin{equation}
\label{eq:privCondition}
 \exists \tilde{\mpre}\in \frU, \tilde{\mpre}' \in \frU^c:  \mathsf{R}^*(\mpre,\bfy) = \mathsf{R}^*(\tilde{\mpre}',\bfy) = \mathsf{R}^*(\tilde{\mpre},\bfy),
\end{equation}
for all preorders $\mpre$ and query-sequences $\bfy$.
That is, a player never returns a sequence of responses that only preferences in $\frU$ induce.
We design protocols where players truthfully respond until they can no longer certify that preferences from both $\frU$ and $\frU^c$ explain the shared responses.
We give algorithms for player $1$, but they easily translate to player $2$.

We consider secrets defined by the set of preorders where $x$ is strictly preferred to $y$, for $x$ and $y$ in $Q$, for example
$
    \frU_{xy} = \{\mpre | x \succ_{\mpre} y\}.
$
The methods we give generalize to conjunctions, e.g., $\frU = \frU_{x_1,y_1} \cap \frU_{x_2,y_2}
$.

\subsection{Privacy by preference-inference}

Player $1$ must ensure that incrementally revealed game information does not 
reveal whether $E_1 \in \frU$.
When player $1$'s preorder is $E$, let
$\bgobjP{1}{q}{E}$ be the objective in the B\"uchi game for semi-automaton state $q$ and let $\lsrP{1}{q}{E}$ be the losing region.
Suppose player $1$ provides losing regions $(\lsrP{1}{\qc_j}{E_1})_{j=1}^{k-1}$ for semi-automaton states $(\qc_j)_{j=1}^{k-1}$, and the mediator queries them for $\lsrP{1}{\qc_k}{E_1}$.
Before revealing $\lsrP{1}{\qc_k}{E_1}$, player $1$ must ensure there exist preorders $E \in \frU$ and $E' \in \frU^c$, that induce these losing regions, \ie \ such that 
\begin{equation}
    \lsrP{1}{\qc_j}{E} = \lsrP{1}{\qc_j}{E'} = \lsrP{1}{\qc_j}{E_1}    \quad \forall j  = 1,\ldots, k.
\end{equation}
If these preorders do not exist, player $1$ must reply $\mathsf{STOP}$, as revealing $\lsrP{1}{\qc_k}{E_1}$ will reveal whether or not $E_1 \in \frU$.

To ensure the $\mathsf{STOP}$ message itself does not reveal whether $E_1 \in \frU$, 
player $1$ can allocate
a pair of \emph{dummy preorders} $\dppl{k} \in \frU$ and $\dpmi{k} \in \frU^c$ on which the protocol stops for the first time at query $k$. 
Defining $\dppl{k}$ and $\dpmi{k}$ ensures that, for any query sequence $\bfy$, $\mathsf{STOP}$ is a feasible response for preferences in both $\frU$ and $\frU^c$.
Players compute new dummy preorders before responding to each query and ensure they are consistent with previous responses. For example, player $1$ should compute $\dppl{k}$ so that
\begin{equation}
    \dppl{k} \in \frU  \text{ and }  \lsrP{1}{q_j}{\dppl{k}} = \lsrP{1}{q_j}{E_1}
       \quad \forall j= 1 , \ldots,  k-1. 
\end{equation}

Algorithm~\ref{alg:mergedAlgo} defines a response function that preserves privacy by (1) ensuring losing regions and punishment strategies do not reveal information, and (2) ensuring $\mathsf{STOP}$ is not disambiguating.
The query comprises a semi-automaton state, and a flag to indicate whether to return a punishment strategy or losing region.
We note that, by Algorithm~\ref{alg:MainAlg}, a punishment flag can only appear on the final query.
On query $k$ the algorithm also takes the previous queries $\qmem{k} = (q_1,\ldots,q_{k-1})$ as well as two lists $\dmempl{k} = (\dppl{1},\ldots, \dppl{k})$ and 
$\dmemmi{k} = (\dpmi{1},\ldots, \dpmi{k})$ of previously generated dummy preorders.
The sets $\qmem{k},\dmempl{k}$ and $\dmemmi{k}$ comprise the memory of Algorithm~\ref{alg:mergedAlgo}.
On the first call to Algorithm~\ref{alg:mergedAlgo},
$\qmem{1} = \emptyset$, $\dmempl{1} = \{\mathcal{E}_{(x,y)}^Q\}$, and
$\dmemmi{1} = \smash{\{\mathcal{E}_\emptyset^Q\}}$.
Algorithm~\ref{alg:mergedAlgo} first checks if $\mpre$ is a dummy preorder, 
and if so, replies $\mathsf{STOP}$.
The algorithm then calls a function \textsc{\genk} to generate preorders that (1) share winning regions, and possibly a punishment strategy, with $\mpre$, and (2) are not in $\dmempl{k} \cup \dmemmi{k}$.
If such preorders exist, a player can respond truthfully without revealing whether $\mpre \in \frU$.
Additionally, a player can use the generated preferences as dummy preferences for the next query.
When queried for a punishment strategy, we assume player $1$ simply returns the whole B\"uchi game  at line~\ref{ln:bucReturn}.

The function \textsc{\genk} must satisfy a series of assumptions that ensure Algorithm~\ref{alg:mergedAlgo} facilitates privacy.

\begin{assumption}
\label{ass:GenKAss}
Let $g = \textsc{\genk}$, and let $\calP^+,\calP^- = g(E, (\qc_j)_{j=1}^k, \frU_{xy}, K, f)$. 
The following conditions hold.
\begin{enumerate}
    \item $\calP^+ \subset \frU_{xy}$, $\calP^- \subset \frU_{xy}^c$, $|\calP^+| \leq K$, and $|\calP^-| \leq K$.
    \item For all $E' \in \calP^+ \cup \calP^-$ we have, for $j = 1,\ldots, k$, 
    \begin{equation}
        (V_r \times \{0\}) \cap \lsrP{1}{q_j}{E} =  (V_r \times \{0\}) \cap  \lsrP{1}{q_j}{E'} .
    \end{equation}
    \item If $f = p$, then $\{\qc_k\}_{\mpre}^\uparrow = \{\qc_k\}_{\mpre'}^\uparrow$.
    \item For $E' \in \calP^+ \cup \calP^-$, $j$ in $\{1,\ldots,k\}$, and $K'$ in $\mathbb{N}$,
    \begin{multline}
        \label{eq:winRegionEqualityCondition}
        g(E, (\qc_l)_{l=1}^j, \frU_{xy},K', f^*_j) = \\ g(E', (\qc_l)_{l=1}^j, \frU_{xy}, K', f^*_j),
    \end{multline}
    where $f^*_j = r$ for $j < k$, and $f^*_j = f$.
\end{enumerate}
Condition 1 requires the secret to be true and false for preferences in $\mathcal{P}^+$ and $\mathcal{P}^-$, respectively.
Conditions 2 and 3 stipulate that the preferences must share losing regions, and possibly punishment strategies, with $E$.
The final condition ensures running Algorithm~\ref{alg:mergedAlgo} on preferences in $\mathcal{P}^+ \cup \mathcal{P}^-$
gives the same outputs as running the algorithm with $E$.

When generating $\calP^+_k$ to compute new dummy preorder $\dpmi{k+1}$, we note that we only need to generate at most $k+1$ preorders, as $\dmempl{k} \cup \dmemmi{k}$ contains at most $k$ preorders in $\frU$. 

\end{assumption}

The following theorem codifies that Algorithm~\ref{alg:mergedAlgo} is privacy-preserving under Assumption~\ref{ass:GenKAss}.

\begin{theorem}
    \label{thm:privacyGuar}
    If Assumption~\ref{ass:GenKAss} holds, Algorithms~\ref{alg:MainAlg} and \ref{alg:mergedAlgo} comprise a privacy-preserving NE synthesis protocol.
\end{theorem}

\textit{Proof sketch\footref{fn:extended_version}: }
When no player returns $\mathsf{STOP}$, the protocol proceeds via Algorithm~\ref{alg:MainAlg}, and is correct by Theorems~\ref{thm:NESuff} and \ref{thm:NENec}. 
We validate privacy.
Let preorder $\mpre \in \frU_{xy}$ and query sequence $(\qc_j,f_j)_{j=1}^k$ be given.
If the response $\rc_k$ to query $k$ is $\mathsf{STOP}$, then by definition of $\dpmi{k}$, and Assumption~\ref{ass:GenKAss}, using Algorithm~\ref{alg:mergedAlgo} with $\mpre$ and $\dpmi{k}$ gives the same output. 
If $\rc_k \neq \mathsf{STOP}$, the same logic implies that  $\mpre$ and $\dpmi{k+1}$ give the same output. Symmetric results hold for $\mpre \in \frU_{xy}^c.$
$\square$

\begin{algorithm}[t]
\caption{A response function for player $1$.}\label{alg:mergedAlgo}
\begin{algorithmic}[1]
\STATE \textbf{Input:} $\frU_{xy}$, $\mpre, \mathcal{U} = (\qc_j)_{j=1}^{k-1}, (\qc_k, \text{flag}), $ \ldots \\ $ \quad \quad \quad \quad \dmempl{k} = (\dppl{j})_{j=1}^k, \dmemmi{k} = (\dpmi{j})_{j=1}^k$
\STATE  \textbf{if} \ $\mpre = \dppl{k}$ or $\mpre = \dpmi{k}$ \textcolor{orange}{\textbf{Reply}} $\mathsf{STOP}$\label{ln:STOPtest}
\STATE $\calP^+_k, \calP^-_k \gets  $
$\textsc{\genk}( 
\mpre, (\qc_j)_{j=1}^k,\frU_{xy}, k+1, \text{flag})$
\STATE  \textbf{if} \ $\calP^+_k \subset \dmempl{k} \cup \dmemmi{k}$ \textcolor{orange}{\textbf{Reply}} $\mathsf{STOP}$ \label{ln:pProcStart}
\STATE  \textbf{if} \ $\calP^-_k \subset \dmempl{k} \cup \dmemmi{k}$
\textcolor{orange}{\textbf{Reply}} $\mathsf{STOP}$ 
\STATE $\dppl{k+1} \gets \dpc'$ for some $\dpc' \in \calP^+_k \setminus ( \dmempl{k}\cup \dmemmi{k} )$
\STATE $\dpmi{k+1} \gets \dpc'$ for some $\dpc' \in \calP^-_k \setminus ( \dmempl{k}\cup \dmemmi{k} )$ \label{ln:pProcEnd}

\IF {flag = $r$}
    \STATE\textcolor{orange}{\textbf{Reply}} with 
    $(V_r \times \{0\}) \cap \lsrP{1}{q_k}{E}$\label{ln:wrReturn}
    \STATE \textcolor{blue}{\textbf{Store}} 
    $\qmem{k+1} = (\qc_j)_{j=1}^k,$\\ $\quad \quad \ \ \dmempl{k+1} = (\dppl{j})_{j=1}^{k+1} , \dmemmi{k+1} = (\dpmi{j})_{j=1}^{k+1} $ \label{ln:updGamma}
\ELSE
    \STATE \textcolor{orange}{\textbf{Reply}} with whole $\mathsf{Buchi}$ game $\mathsf{Buchi}(\bgobjP{1}{q_k}{E})$ \label{ln:bucReturn}
\ENDIF
\end{algorithmic}
\end{algorithm}

\subsection*{Implementing \textsc{\genk}.}

We generate preorders that satisfy a set of constraints on strict preference.
For $u,v$ in $Q$, a constraint is a tuple $(u,v,\succ)$ or $(u,v,\not\succ)$.
A preorder $\mpre$ satisfies $(u,v,\succ)$ if and only if $u \succ_\mpre v$, and satisfies $(u,v,\not\succ)$ if and only if $u \not\succ_\mpre v$. 
The \textsc{CheckConstraints} procedure in Algorithm~\ref{alg:MainBodyPrefGeneration}
returns preorders that satisfy a set $\calC$ of constraints, %
by iteratively adding edges to a preorder $\calE$ until it satisfies constraints or we get a contradiction.
\textsc{\genfc} uses \textsc{CheckConstraints} to generate sets of these preorders.
Algorithm~\ref{alg:MainBodyPrefGeneration} is correct by Lemmas~\ref{lem:prefGenCorr1} and \ref{lem:prefGenCorr2}.
Lemma~\ref{lem:prefGenCorr2} additionally codifies \textsc{\genfc}'s completeness, as it always returns $M$ preorders if $M$ such preorders exist.

\begin{lemma}
    \label{lem:prefGenCorr1}
    $\textsc{CheckConstraints}(Q, \mathcal{C}, \mathcal{E}_{init})$ returns the smallest preorder satisfying constraints $\mathcal{C}$ that contains $\mathcal{E}_{init}$
\end{lemma}

\textit{Proof sketch\footref{fn:extended_version}:}
Let $\calH$, a preorder containing $\mathcal{E}_{init}$ and satisfying $\mathcal{C}$, be arbitrary.
Every edge that we add to $\mathcal{E}$, on lines \ref{ln:pa1}, \ref{ln:pa2}, and \ref{ln:pa3}, lies in $\calH$.
If we return $\calE$ it is a preorder and satisfies $\calE \subseteq \calH$, and satisfies $\calC$.
If we return $\mathsf{Fail}$, the existence of $\calH$ leads to a contradiction. $\square$

\begin{lemma}
    \label{lem:prefGenCorr2}
    Let
    $ \mathcal{P} = \textsc{\genfc}(Q,\mathcal{C},M)$.
    Either $|\mathcal{P}| = M$ or $\mathcal{P}$ contains all preorders satisfying $\mathcal{C}$.
\end{lemma}

\textit{Proof sketch\footref{fn:extended_version}: }
We show $|\calP|$ to increase when it does not contain all preorders satisfying $\calC$.
If there exists a preorder $E'$ outside $\calP$ that satisfies  $\calC$, then we can find $E$ in $ \calP$ and an edge $e$ such that 
$E \cup \{e\} \subset E'$, but $ E \cup \{e\} \not\subset B$ for any $B$ in $ \calP$. Thus, $\textsc{CheckConstraints}(Q, \mathcal{C}, E \cup \{e\})$ generates a new preorder, by Lemma~\ref{lem:prefGenCorr1}, and $\calP$ keeps growing.
$\square$

\begin{algorithm}
\caption{Generating preorders that satisfy constraints.}
\label{alg:MainBodyPrefGeneration}
\begin{algorithmic}[1]
\STATE \textsc{\genfc}($Q, \mathcal{C}, K$)
\STATE $\mathcal{E} \gets \textsc{CheckConstraints}(Q,\mathcal{C},\emptyset)$
\STATE \textbf{if} $\calE = \mathsf{Fail}$,  \textbf{return} $\emptyset$
\STATE \npvar $ \ \gets $ \textbf{true}, $\calP \gets \{\calE\}$
\WHILE {$|\mathcal{P}| < K$, \npvar \ is \textbf{true}}
    \STATE \npvar \ $\gets$ \textbf{false}
    \FOR {$e \in Q^2, E \in \mathcal{P}$}
        \STATE $\mathcal{E}' \gets \textsc{CheckConstraints}(Q,\mathcal{C},E \cup e)$
        \IF{$\calE' \neq \mathsf{Fail}$ and $\mathcal{E}'$ not in $\mathcal{P}$ }
            \STATE $\mathcal{P} \gets \mathcal{P} \cup \{\mathcal{E}'\}$,  \npvar \ $\gets$ \textbf{true}
        \ENDIF
    \ENDFOR
\ENDWHILE
\STATE \textbf{return} first $K$ elements of $\calP$
\vspace{0.4\baselineskip}
\STATE \textsc{CheckConstraints}($Q, \mathcal{C}, \mathcal{E}_{init}$)
    \STATE $\mathcal{E} \gets \mathcal{E}_{init} \cup \{(q,q) | q \in Q\} \cup \{(u,v) | (u,v,\succ) \in \mathcal{C}\}$ \label{ln:pa1}
    \WHILE{\textbf{true}}
        \STATE $\mathcal{E} \gets \textsc{TransitiveClosure}(\mathcal{E})$, 
        $\text{Flag} \gets \mathsf{Valid}$ \label{ln:pa2}
        \STATE \textbf{if} {$(v,u) \in \mathcal{E}$ for some $(u,v,\succ) \in \mathcal{C}$} \textbf{return} $\mathsf{Fail}$ \label{ln:contCheck}
        \IF {$(u,v) \in \mathcal{E} \wedge (v,u) \notin \calE$ for some $(u,v,\not\succ) \in \mathcal{C}$}
            \STATE $\mathcal{E} \gets \mathcal{E} \cup (v,u)$,  $\text{Flag} \gets \mathsf{Invalid}$ \label{ln:pa3}
        \ENDIF
        \STATE \textbf{if} {$\text{Flag} = \mathsf{Valid}$} \textbf{return} $\mathcal{E}$ \label{ln:valCheck}
    \ENDWHILE    
\end{algorithmic}
\end{algorithm}

Algorithm~\ref{alg:genKImp} implements \textsc{\genk} by using
Algorithm~\ref{alg:MainBodyPrefGeneration} 
to generate preorders $E'$ that share upper closures with a reference preorder $\mpre$. 
If $\upc{\qc}{E} = \upc{q}{E'}$, then preorders $E$ and $E'$ induce the same B\"uchi objective $\bgobj{1}{q}$ for semi-automaton state $q$, and thus the algorithm satisfies conditions 2 and 3 in Assumption~\ref{ass:GenKAss}.
As Algorithm~\ref{alg:genKImp} only depends on a preorder through the upper closures, and the returned preorders share upper closures, it will also satisfy condition 4 in Assumption~\ref{ass:GenKAss}.
Combining Algorithms~\ref{alg:mergedAlgo} and \ref{alg:genKImp} gives a privacy-preserving response function, and we term
 this method of ensuring privacy as \emph{upper-closure-privacy}.
 While \emph{upper-closure-privacy} ensures privacy, it does not use the obfuscation that the game's structure provides.

\begin{algorithm}
\caption{An implementation of \textsc{\genk}.}
\label{alg:genKImp}
\begin{algorithmic}[1]
\STATE $\textsc{\genk}(E, (\qc_1,\ldots,\qc_k), \frU_{xy}, K, f)$
\STATE $\calC \gets \bigcup_{j=1}^k \{(v,\qc_j,\succ) | v \in \{\qc_j\}^\uparrow_{\mpre}\}$ 
\STATE $\calC \gets \calC \cup \bigcup_{j=1}^k \{(v,\qc_j,\not\succ) | v \not\in \{\qc_j\}^\uparrow_{\mpre}\}$ 
\STATE $\mathcal{C}^+ \gets \calC \cup \{(x,y,\succ)\}$ \ $\mathcal{C}^- \gets \calC \cup \{(x,y,\not\succ)\}$
\STATE $\mathcal{P}^+ \gets $ \textsc{\genfc}($Q,\mathcal{C}^+,K$)
\STATE $\mathcal{P}^- \gets $ \textsc{\genfc}($Q,\mathcal{C}^-,K$)
\end{algorithmic}
\end{algorithm}

Given losing regions $(\lsrP{1}{\qc_j}{E})_{j=1}^k$, 
\ie, complements of winning regions, 
we can search for preorders $E'$ such that $\lsrP{1}{\qc_j}{E'} = \lsrP{1}{\qc_j}{E}$ for all $j$, and we term this method \emph{winning-region-privacy}.
As this method is only a function of losing regions and returns preorders that share losing regions, it will satisfy condition 4 in Assumption~\ref{ass:GenKAss}.
We implement \emph{winning-region-privacy} for a small game in the next section.

\section{Experimental results}

We evaluate privacy with a shared-delivery scenario involving a drone and a truck. The drone can not traverse areas with dense tree cover, while the truck struggles on poor-quality roads, requiring the vehicles to exchange the package to complete delivery. 
The vehicles have different preferences over the final location of the package, so each vehicle must carefully consider when to exchange the package. 
We synthesize Nash equilibria in this scenario and explore whether we can synthesize equilibria while ensuring privacy.

We model this scenario as a game on a graph. A network $\mathcal{G} = (\mathcal{T}, \mathcal{F})$ defines the game graph, and for each player $i \in \{1, 2\}$, a set $\mathcal{F}_i \subseteq \mathcal{F}$ defines where their vehicle can transport the package. 
The game state is a tuple $(t, \iota)$, where $t \in \mathcal{T}$ is the package's location, and $\iota$ indicates the vehicle holding it.
We set $AP = \mathcal{T}$ where $L(t,\iota)$ is the set $\{t\}$.

The players have preferences on the final location of the package, which we encode via a preference automaton tracking this location.
For a set $\mathcal{D}\subset \calT$, we define the semi-automaton
$
    \langle \calD \cup \{\top\}, 2^\mathcal{T}, \delta, \bot\rangle,
$
where, for all $q$,
$\delta(q,\{t\}) = t$ for $t \in \mathcal{D}$ and $\delta(q,\{t\}) = \top$ for $t$ not in $\mathcal{D}$.

\subsection{Synthesizing Nash equilibria without violating privacy}

We first give an example of a game where we can compute non-trivial Nash equilibria without violating privacy using upper-closure-privacy (UC-privacy).
Figure~\ref{fig:envExample} shows an environment with $\mathcal{D} = \{a,b,c,d\}$. 
The players have the following preferences, up to transitive closure.
\begin{align*}
    \begin{matrix}
         a \ \succ_{E_1} \  d, 
        & 
         d \ \succ_{E_1} \  b,
        &
        d \ \succ_{E_1} \  c,
        &
        a,b,c,d \ \succ_{E_1} \top.
        \\
        b \ \succ_{E_2} \  d,
        &
        d \ \succ_{E_2} \  a,
        &
         d \ \succ_{E_2} \  c,
         &
         a,b,c,d \ \succ_{E_2} \top.
        \normalcolor
    \end{matrix}
\end{align*}
The secret is given by $\frU_{ac} = \{E| a \succ_E c\}$.

There exists a query sequence that allows the players to synthesize the Nash equilibrium (NE) shown in Figure~\ref{fig:prefExplanations} without revealing the secret.
Suppose players exchange winning regions corresponding to semi-automaton states in the order $(\top, b, d )$.
A NE, shown in Figure~\ref{fig:prefExplanations}, exists terminating with the semi-automaton state at $d$, and these queries find this equilibrium.
The exchanged information is not enough to recover whether $\mpre_1 \in \frU$ or $\mpre_1 \in \frU^c$.
Figure~\ref{fig:prefExplanations} shows two preference relations for player $1$, one in $\frU$ and one in $\frU^c$,  that produce the exchanged messages.

\begin{figure}
    \centering
    \begin{minipage}{0.19\textwidth}
        \includegraphics[width=\linewidth]{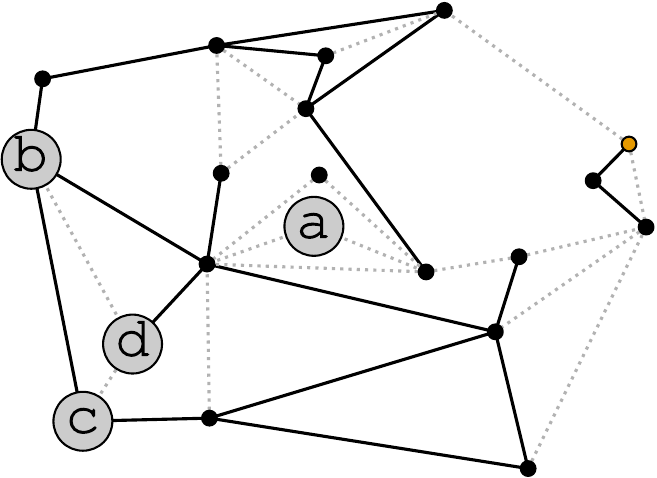}
        \label{fig:envExamplePlayer1}
    \end{minipage}
    \hfill
    \begin{minipage}{0.19\textwidth}
        \includegraphics[width=\linewidth]{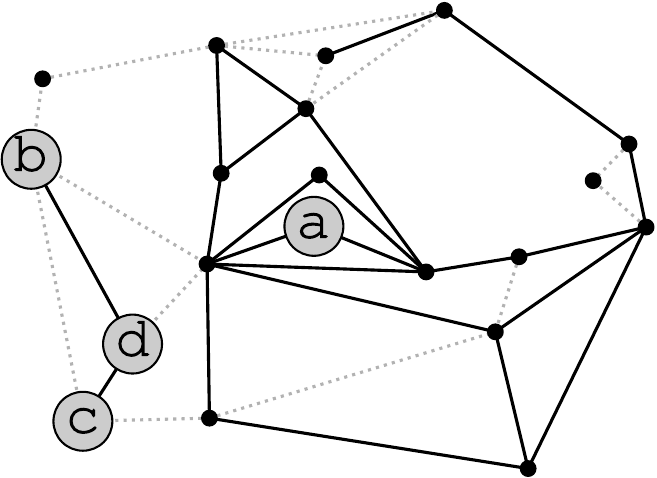}
        \label{fig:envExamplePlayer2}
    \end{minipage}
    \caption{An example environment. Gray nodes are in $\calD$, and the initial node is orange. 
    The left graph is $(\mathcal{T},\mathcal{F}_1)$, and the right graph is $(\mathcal{T},\mathcal{F}_2)$
    }
    \label{fig:envExample}
\end{figure}

\begin{figure}[ht]
\centering
\begin{minipage}{0.19\textwidth}
    \centering
    \includegraphics[width=\textwidth]{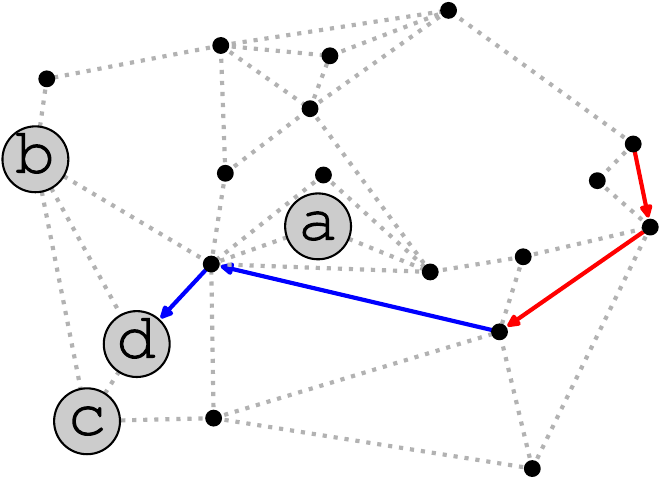}
    (a)
    \label{fig:left}
\end{minipage}%
\hspace{30pt}
\begin{minipage}{0.12\textwidth}
    \centering
    \includegraphics[width=0.7\textwidth]{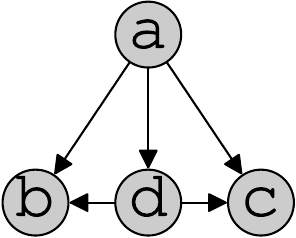}\ (b)
    \vspace{5pt}
    \label{fig:top}
    \includegraphics[width=0.7\textwidth]{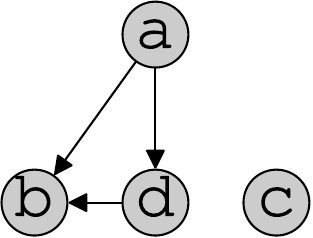}\ (c)
    \label{fig:bottom}
\end{minipage}
    \caption{
    (a) Nominal paths of NE strategies in the game in Figure~\ref{fig:envExample}. A blue and red edge indicates that players $1$ and $2$ are in control respectively.
    (b,c) Preorders $\mpre_1 \in \frU$ and  $\tilde{\mpre}_1 \in \frU^c$  for player 1 that produce the same query responses. These preorders satisfy $\{v\}^\uparrow_{\mpre_1} = \{v\}^\uparrow_{\tilde{\mpre}_1}$ for $v$ in $\{\top,b,d\}$. }
    \label{fig:prefExplanations}
\end{figure}

\subsection{Permissiveness of privacy checks}
\label{sec:privPerm}

The above game only gives an example of a protocol run that results in \ac{ne} synthesis.
We now measure whether privacy inhibits \ac{ne} synthesis by estimating the number of queries before player $1$ stops the protocol, for randomly sampled preferences and query sequences.

\begin{table}[h]
\centering
\caption{Probability that player $1$ replies $\mathsf{STOP}$ on the $k^{\text{th}}$ query.}
\begin{tabular}{|c|c|c|c|c|c|c|}
\hline
k& $1$ & $2$ & $3$ & $4-6$ & $7-9$ &$\geq 9$ \\
\hline
UC-privacy & 0.079 & 0.084 & 0.064 & 0.137 & 0.636 & 0.0 \\
\hline
WR-privacy & 0.0 & 0.064 & 0.113 & 0.428 & 0.392 & 0.003 \\
\hline
\end{tabular}
\label{tab:successRate}
\end{table}

We consider a game where $(\mathcal{T},\mathcal{F})$ is a $3 \times 3$ grid and $(\mathcal{F}_1,\mathcal{F}_2)$ are such that players $1$ and $2$  can move along rows and columns respectively.
The players have preferences on nodes in $\calT$, \ie, $\calD = \calT$.
We set $\frU_{xy}$ such that $x = (0,0)$, and $y = (0,1)$, \ie, player $1$ hides their preference on two nodes in the first row.
Given a reference losing region $\tilde{U}$, we derive necessary and sufficient conditions on upper closures $\upc{\qc}{E'}$ of a preorder $E'$ so that $\lsrP{1}{\qc}{E'} = \tilde{U}$, and we use these conditions to implement 
{winning-region-privacy} (WR-privacy).\footref{fn:extended_version}
Given %
losing regions $(\lsrP{1}{\qc_j}{E})_{j=1}^k$,
we compute upper closure constraints $(Z_j)_{j=1}^k$ that satisfy these conditions. 
We then Algorithm~\ref{alg:MainBodyPrefGeneration} to find preorders $E'$ such that $\upc{\qc_j}{E'} = Z_j$ for all $j$ in $1,\ldots,k$.
This method requires searching over the set of constraint sequences, and we limit the search to a fixed number of constraint sequences.

We sample $1000$ preorders by sampling directed trees and taking transitive closures.
For each preorder, we sample a query sequence by shuffling $Q$, and we place $(x,y)$ at the end of the query sequence to avoid trivial protocol terminations.
Reordering query-sequences as above is reasonable if the mediator knows $x$ and $y$.
Note that this assumption does not imply that the mediator knows whether or not $E_1 \in \frU_{xy}$. 

Table~\ref{tab:successRate} confirms that {upper-closure-privacy} allows a mediator to query a player for multiple winning regions without violating privacy.
Indeed, over half of the sampled preferences allow up to six queries. 
However, we notice that, on some preference-query pairs, the protocol stops instantly.

As winning-region-privacy involves a search over a larger space of preorders to certify privacy,  it reduces the probability that a single query leads to a $\mathsf{STOP}$ reply.
However, winning-region-privacy is limited for extended query sequences due to the need to search over sets of constraints.

\section{Conclusion}

We developed an algorithm for Nash equilibrium synthesis in games on graphs where players have partial preferences over temporal goals.
We studied a setting where players' preferences are private, but they communicate to synthesize equilibria. 
We proved the algorithm's correctness and proved that the algorithm preserved the privacy of the players' preferences. 
We experimentally validated that we could synthesize non-trivial Nash equilibria while maintaining privacy.

\addtolength{\textheight}{-12cm}   %

\bibliography{GogSansPref}

\begin{thebibliography}{10}

\bibitem{andersson_acyclicity_2010}
Daniel Andersson, Vladimir Gurvich, and Thomas~Dueholm Hansen.
\newblock On acyclicity of games with cycles.
\newblock {\em Discrete Applied Mathematics}, 158(10):1049--1063, 2010.

\bibitem{bade_nash_2005}
Sophie Bade.
\newblock Nash equilibrium in games with incomplete preferences.
\newblock {\em Economic Theory}, 26(2):309--332, 2005.

\bibitem{baier2008principles}
Christel Baier and Joost-Pieter Katoen.
\newblock {\em Principles of model checking}.
\newblock MIT press, 2008.

\bibitem{bouyer_pure_2015}
Patricia Bouyer, Romain Brenguier, Nicolas Markey, and Michael Ummels.
\newblock Pure {N}ash equilibria in concurrent deterministic games.
\newblock {\em Logical Methods in Computer Science}, Volume 11, Issue 2, 2015.

\bibitem{bouyer2016stochastic}
Patricia Bouyer, Nicolas Markey, and Daniel Stan.
\newblock Stochastic equilibria under imprecise deviations in terminal-reward concurrent games.
\newblock {\em arXiv:1609.04089}, 2016.

\bibitem{brihaye_multiplayer_2012}
Thomas Brihaye, Julie De~Pril, and Sven Schewe.
\newblock Multiplayer cost games with simple {N}ash equilibria.
\newblock In {\em Logical Foundations of Computer Science}, pages 59--73, 2013.

\bibitem{bruyere2025games}
V{\'e}ronique Bruy{\`e}re, Christophe Grandmont, and Jean-Fran{\c{c}}ois Raskin.
\newblock Games with $\omega$-automatic preference relations.
\newblock {\em arXiv:2503.04759}, 2025.

\bibitem{bruyere2017existence}
V{\'e}ronique Bruy{\`e}re, St{\'e}phane Le~Roux, Arno Pauly, and Jean-Fran{\c{c}}ois Raskin.
\newblock On the existence of weak subgame perfect equilibria.
\newblock In {\em Foundations of Software Science and Computation Structures}, pages 145--161. Springer, 2017.

\bibitem{chen2025differentially}
Ying Chen, Qian Ma, Peng Jin, and Shengyuan Xu.
\newblock Differentially private distributed {N}ash equilibrium seeking for aggregative games with linear convergence.
\newblock {\em IEEE Transactions on Fuzzy Systems}, 33(6):1853--1863, 2025.

\bibitem{cummings_privacy_2015}
Rachel Cummings, Michael Kearns, Aaron Roth, and Zhiwei~Steven Wu.
\newblock Privacy and truthful equilibrium selection for aggregative games.
\newblock In {\em Web and Internet Economics.}, volume 9470, pages 386--299, 2015.

\bibitem{fijalkow2023games}
Nathana{\"e}l Fijalkow, Nathalie Bertrand, Patricia Bouyer-Decitre, Romain Brenguier, Arnaud Carayol, John Fearnley, Hugo Gimbert, Florian Horn, Rasmus Ibsen-Jensen, Nicolas Markey, et~al.
\newblock Games on graphs.
\newblock {\em arXiv:2305.10546}, 2023.

\bibitem{fu2015concurrent}
Jie Fu, Herbert~G. Tanner, and Jeffrey Heinz.
\newblock Concurrent multi-agent systems with temporal logic objectives: game theoretic analysis and planning through negotiation.
\newblock {\em IET Control Theory and Applications}, 9(3):465--474, 2015.

\bibitem{giacomo_linear_nodate}
Giuseppe~De Giacomo and Moshe~Y Vardi.
\newblock Linear {Temporal} {Logic} and {Linear} {Dynamic} {Logic} on {Finite} {Traces}.

\bibitem{gottlieb2001prufer}
Jens Gottlieb, Bryant~A Julstrom, G{\"u}nther~R Raidl, and Franz Rothlauf.
\newblock Pr{\"u}fer numbers: A poor representation of spanning trees for evolutionary search.
\newblock In {\em Annual Conference on Genetic and Evolutionary Computation}, pages 343--350, 2001.

\bibitem{guelev2025qualitative}
Dimitar~P Guelev.
\newblock On qualitative preference in alternating-time temporal logic with strategy contexts.
\newblock {\em arXiv:2502.13436}, 2025.

\bibitem{harsanyi1967games}
John~C Harsanyi.
\newblock Games with incomplete information played by “{B}ayesian” players, i--iii part i. the basic model.
\newblock {\em Management science}, 14(3):159--182, 1967.

\bibitem{huo_compression-based_2024}
Wei Huo, Xiaomeng Chen, Kemi Ding, Subhrakanti Dey, and Ling Shi.
\newblock Compression-based privacy preservation for distributed {N}ash equilibrium synthesis.
\newblock {\em IEEE Control Systems Letters}, 8:886--891, 2024.

\bibitem{kearns_private_2012}
Michael Kearns, Mallesh~M. Pai, Aaron Roth, and Jonathan Ullman.
\newblock Private equilibrium release, large games, and no-regret learning.
\newblock {\em arXiv:1207.4084}, 2012.

\bibitem{kulkarni_nash_2024}
Abhishek~N. Kulkarni, Jie Fu, and Ufuk Topcu.
\newblock Nash equilibrium in games on graphs with incomplete preferences.
\newblock {\em arXiv:2408.02860}, 2024.

\bibitem{noauthor_sequential_2024}
Abhishek~N. Kulkarni, Jie Fu, and Ufuk Topcu.
\newblock Sequential decision making in stochastic games with incomplete preferences over temporal objectives.
\newblock In {\em AAAI Conference on Artificial Intelligence}, 2025.

\bibitem{le2013infinite}
St{\'e}phane Le~Roux.
\newblock Infinite sequential {N}ash equilibrium.
\newblock {\em Logical Methods in Computer Science}, 9(2), 2013.

\bibitem{le2014winning}
St{\'e}phane Le~Roux.
\newblock From winning strategy to {N}ash equilibrium.
\newblock {\em Mathematical Logic Quarterly}, 60(4-5):354--371, 2014.

\bibitem{le2018extending}
St{\'e}phane Le~Roux and Arno Pauly.
\newblock Extending finite-memory determinacy to multi-player games.
\newblock {\em Information and Computation}, 261(4):676--694, 2016.

\bibitem{li_distributed_1987}
Shu Li and Tamer Başar.
\newblock Distributed algorithms for the computation of noncooperative equilibria.
\newblock {\em Automatica}, 23(4):523--533, July 1987.

\bibitem{lin_statistical_2024}
Yeming Lin, Kun Liu, Dongyu Han, and Yuanqing Xia.
\newblock Statistical privacy-preserving online distributed {N}ash equilibrium tracking in aggregative games.
\newblock {\em IEEE Transactions on Automatic Control}, 69(1):323--330, 2024.

\bibitem{nakanishi2025strategies}
Rindo Nakanishi, Yoshiaki Takata, and Hiroyuki Seki.
\newblock Strategies and equilibria on indistinguishability of winning objectives and related decision problems.
\newblock {\em IEICE Transactions on Information and Systems}, 108(3):239--251, 2025.

\bibitem{patek1999stochastic}
Stephen~D. Patek and Dimitri~P. Bertsekas.
\newblock Stochastic shortest path games.
\newblock {\em SIAM Journal on Control and Optimization}, 37(3):804--824, 1999.

\bibitem{rahmani_probabilistic_2023}
Hazhar Rahmani, Abhishek~N. Kulkarni, and Jie Fu.
\newblock Probabilistic planning with partially ordered preferences over temporal goals.
\newblock {\em arXiv:2209.12267}, 2023.

\bibitem{rogers_asymptotically_2014}
Ryan~M. Rogers and Aaron Roth.
\newblock Asymptotically truthful equilibrium selection in large congestion games.
\newblock In {\em Economics and Computation}, pages 771--782, Palo Alto California USA, June 2014. ACM.

\bibitem{rozen_games_2018}
V.~V. Rozen.
\newblock Games with ordered outcomes.
\newblock {\em Journal of Mathematical Sciences}, 235(6):740--755, 2018.

\bibitem{wang_differentially_2024-1}
Yongqiang Wang and Angelia Nedić.
\newblock Differentially private distributed algorithms for aggregative games with guaranteed convergence.
\newblock {\em IEEE Transactions on Automatic Control}, 69(8):5168--5183, 2024.

\bibitem{ye_differentially_2022}
Maojiao Ye, Guoqiang Hu, Lihua Xie, and Shengyuan Xu.
\newblock Differentially private distributed {N}ash equilibrium seeking for aggregative games.
\newblock {\em IEEE Transactions on Automatic Control}, 67(5):2451--2458, 2022.

\bibitem{zanardi_posetal_2022}
Alessandro Zanardi, Gioele Zardini, Sirish Srinivasan, Saverio Bolognani, Andrea Censi, Florian Dörfler, and Emilio Frazzoli.
\newblock Posetal games: efficiency, existence, and refinement of equilibria in games with prioritized metrics.
\newblock {\em IEEE Robotics and Automation Letters}, 7(2):1292--1299, 2022.

\end{thebibliography}

\addtolength{\textheight}{12cm}
\clearpage

\section*{APPENDIX}

We first elaborate on the proof of the termination property of proper strategies.

\begin{proof_alt}[Proof of Propostion~\ref{prop:ProperPolicyTerm}]
    Suppose that $ \rho = \mathsf{Path}_{\gterm{G}}((s_0,0),{\pi_1},{\pi_2})$ never visits $S\times\{1\}$. It can not hold that there exists $N_1$ and $N_2$ such that $\rho[n] \in S_1\times\{0\}$ for all $n \geq N_1$ and $\rho[n] \in S_2\times\{0\}$ for all $n \geq N_2$. Hence, for one of $i\in \{1,2\}$, there must exist $n$ such that $\pi_i(\rho[0]\ldots\rho[n]) = \tau$. However, the definition of the terminating game then implies that $\rho[n+1] \in S\times\{1\}$, a contradiction.
\end{proof_alt}

\section{Nash synthesis proofs}

We give here proof of the correctness of Algorithm~\ref{alg:MainAlg}. We first show that when Algorithm~\ref{alg:MainAlg} returns a strategy pair, said strategy pair is a Nash equilibrium. We then prove that when Algorithm~\ref{alg:MainAlg} returns \textbf{false}, a Nash equilibrium does not exist.
Throughout this section, the preference automaton preorders $E_1$ and $E_2$ are fixed and implicit.

We first describe extra notation.
For a set $A$, $A^\omega$ is the set of infinite sequences on $A$.
The set of paths in the game $\gterm{H}$ is $\PathSet{\gterm{H}}$.
When players $1$ and $2$ follow strategies $\pi_1$ and $\pi_2$ respectively from an initial history $z_0\ldots z_k$ in the product game $\gterm{H}$, we denote the resulting path by
$\PathsInd{\gterm{H}}{z_0\ldots z_k}{\pi_1}{\pi_2}$.
The set of strategies for player $i$ in the product game $\gterm{H}$
is $\StratSpace{i}{\gterm{H}}$, for $i \in \{1,2\}$.
For two finite sequences $\rho$ and $\rho'$, the maximal common prefix is the longest prefix that is common to both paths.

\subsubsection{Sufficiency of Algorithm~\ref{alg:MainAlg}}
We first prove that we can construct proper strategies from strategies that Algorithm~\ref{alg:MainAlg} returns. 
Given a memoryless winning strategy $\sigma_{-2}$ for player $1$ in player $2$'s B\"uchi game, we show that we may construct a new winning strategy, 
\begin{equation}
    \hat{\sigma}_{-2}(v,0) = M_2({\sigma}_{-2})(v,0) = 
    \begin{cases}
       \tau & \text{if } v \in V_2^-(q), \\ 
       \sigma_{-2}(v,0) & \text{otherwise}
    \end{cases}
\end{equation}
which is also winning. We additionally prove that this strategy ensures the game visits no vertex in $V_1\times\{0\}$ more than once.

\begin{lemma}[Restatement of Lemma~\ref{lem:noRepMain}]
\label{lem:noRep}
     If $\sigma_{-2}$ is a punishment strategy in $\mathsf{Buchi}(\bgobj{2}{q})$, 
    then $\hat{\sigma}_{-2} = M_2(\sigma_{-2})$ is also a punishment strategy.
    Also, for any $u\in \lsr{2}{q}$ and $\pi_2 \in \Pi_2^{\gterm{H}}$, $\rho = \PathsInd{\gterm{H}}{u}{\hat{\sigma}_{-2}}{\pi_2}$ visits each vertex in $V_1\times\{0\}$ at most once.
\end{lemma}
\begin{proof}
    We first prove that $\hat{\sigma}_{-2}$ is still winning from all nodes in $U_2(q)$. Indeed, let $\rho\in \PathSet{\gterm{H}}$ be a path consistent with $\hat{\sigma}_{-2}$, which starts in $U_2(q)$. Consider two cases.

    \textit{Case 1: } Assume  $\rho$ does not visit $(V^-_2(q)\cap V_1) \times\{0\}$. On this path $\hat{\sigma}_{-2}(h) = \sigma_{-2}(h)$ where $h$ is a prefix of $\rho$. As ${\sigma}_{-2}$ is winning from $u$, $\rho$ satisfies player $1$'s objective in player $2$'s B\"uchi game.

    \textit{Case 2: } Assume $\rho$ does visit $(V^-_2(q)\cap V_1) \times\{ 0\}$ at some time-step $n$. Then
    \begin{equation}
        \rho[n+1] = \Delta(\rho[n], \hat{\sigma}_{-2}(\rho[n])) = \Delta(\rho[n], \tau).
    \end{equation}
    As $\rho[n] \in V^-_2(q) \times\{ 0\}$, we have $\rho[n+1] \in V^-_2(q) \times\{ 1\}$. Hence, this path again satisfies player $1$'s objective in player $2$'s B\"uchi game.

    We next prove that 
    $\PathsInd{\gterm{H}}{u}{\hat{\sigma}_{-2}}{\pi_2}$
    visits each vertex in $V_1\times\{0\}$ at most once for any $\pi_2 \in \Pi_2^{\gterm{H}}$. Suppose by contradiction that there exists $\pi_2'$ such that $\PathsInd{\gterm{H}}{u}{\hat{\sigma}_{-2}}{\pi_2'}$ repeats some vertex $(v,0)$ in $V_1 \times \{0\}$. That is,
    \begin{equation}
        \PathsInd{\gterm{H}}{u}{\hat{\sigma}_{-2}}{\pi_2'} = u \cdot h \cdot (v,0) \cdot c \cdot (v,0) \cdot \rho,
    \end{equation}
    where $h\in (V\times\{0\})^*$, $c\in (V\times\{0\})^*$, $v\in V_1$ and $\rho \in (V\times\{0,1\})^\omega$. 
    Because $\hat{\sigma}_{-2}$ is memoryless, there exists a strategy $\pi_2''\in \Pi_2^{\gterm{H}}$ such that 
    \begin{equation}
        \PathsInd{\gterm{H}}{u}{\hat{\sigma}_{-2}}{\pi_2''}
         = u\cdot h\cdot(v,0)\cdot c\cdot (v,0)\cdot c\cdot (v,0),\ldots.
    \end{equation}

    As $\hat{\sigma}_{-2}$ is winning from $u$ for player $1$, it then must hold that all nodes in $(v,0)\cdot c$ lie in $V_2^-(q)$. 
    However, it then follows that $(v,0)$ lies in $V_1 \cap V_2^-(q)$ and so, by definition of $\hat{\sigma}_{-2}$, player $1$ will play $\tau$, which contradicts the fact that the path does not terminate.
    $\square$
\end{proof}

In the sequel, we assume, without loss of generality, that $\sigma_{-2}$ is of the form $M_2(\sigma_{-2})$.

We now give a formal definition of the strategy construction procedure. We leverage Lemma~\ref{lem:noRep} to ensure that we can make the final strategy proper without introducing opportunities for players to deviate to outcomes they strictly prefer.
For player $1$, the strategy construction procedure is given by
\begin{multline}
    C_1(\mathbf{u}, \sigma_{-2})(h) 
    = \\
    \begin{cases}
        a & \text{if } u_0 \ldots u_k = h,
        T(a,u_k) = u_{k+1}, \\
        \sigma_{-2}(v_k) & \text{if  $h = u_0 \ldots u_l v_{l+1} \ldots v_k$, where} \\
         & \text{$u_0,\ldots,u_l$ is a maximal common}\\  
         & \text{prefix of $\mathbf{u}$ and $h$, and} \\ 
         & \text{$v_{l+1},\ldots,v_k$ repeats no node in $V_1\times\{0\}$},\\
         \tau & \text{otherwise. }
    \end{cases}
\end{multline}
Informally, the first clause of this construction defines the nominal behavior. The second clause ensures player $1$ punishes player $2$ if player $2$ deviates. The final clause ensures that the strategy terminates when any node in $V_1 \times \{0\}$ repeats after the initial deviation.

We now prove that this strategy is proper.
\begin{proposition}[Restatement of Proposition~\ref{prop:ProperCombination}]
    \label{prop:ProperCombinationApp}
    The strategy $C_1(\mathbf{u}, \sigma_{-2})$ is proper.
\end{proposition}
\begin{proof}
    If a path $\rho\in \PathSet{\gterm{H}}$ terminates, then it trivially satisfies one of the conditions in the proper policy definition.
    Instead, let $\rho\in \PathSet{\gterm{H}}$ be a path consistent with $C_1(\mathbf{u}, \sigma_{-2})$ that never visits $V\times\{1\}$. This path must eventually deviate from $\mathbf{u}$. 
    If $\rho$ visits $V_1\times\{0\}$ infinitely often, then some node in $V_1\times \{0\}$ must repeat after the initial deviation. However, by definition of $C_1$ at this node, player $1$ will take action $\tau$, contradicting the fact that the path visits $V_1\times\{0\}$ infinitely often.
    Hence there exists $N$ such that $n \geq N \implies \rho[n] \in V_2\times\{0\}$, and the path satisfies one of the proper strategy conditions.
    $\square$
\end{proof}

We remark that Lemma~\ref{lem:noRep} and Proposition~\ref{prop:ProperCombination} have analogous results when we reverse the players.

Finally, we show that the strategies form a Nash equilibrium strategy pair.
\begin{theorem}
    If the algorithm terminates for final state $w\in V$, the proper strategy pair $(C_1(\mathbf{u}, \sigma_{-2}),C_2(\mathbf{u}, \sigma_{-1}))$ forms a Nash equilibrium. 
\end{theorem}
\begin{proof}
    Let $\mathbf{u} = u_0 u_1 \ldots u_l$ where $u_0 = (v_0,0)$ and $u_l = (w,1)$.
    $\PathsInd{\gterm{H}}{(v_0,0)}{C_1(\mathbf{u}, \sigma_{-2})}{C_2(\mathbf{u}, \sigma_{-1})}$ reaches $(w,1)$ trivially.

    Suppose player $2$ instead follows a different proper strategy $\pi_2$, such that $\rho = 
    \PathsInd{\gterm{H}}{(v_0,0)}{C_1(\mathbf{u}, \sigma_{-2})}{\pi_2}$ deviates from $\mathbf{u}$ at some time-step $k$. That is, $\rho[k] \in V_2\times\{0\}$, $\rho[k] = u_k$ and $\rho[k+1] \neq u_{k+1}$. By definition of $\mathbf{u}$, $u_k$ lies in the losing region $U_2(q)$. We may re-express the path $\rho$ as 
    \begin{equation}
        \rho = u_0 \ldots u_{k-1},\PathsInd{\gterm{H}}{u_0 \ldots u_k}{C_1(\mathbf{u}, \sigma_{-2})}{\pi_2}.
    \end{equation}
    In $\PathsInd{\gterm{H}}{u_0 \ldots u_k}{C_1(\mathbf{u}, \sigma_{-2})}{\pi_2}$, player $1$ will follow $\sigma_{-2}$ until a node in $V_1\times\{0\}$ repeats, but by Lemma~\ref{lem:noRep}, this repetition never occurs.
    Thus, we have 
    \begin{equation}
        \rho = u_0 \ldots u_{k-1} \cdot 
        \PathsInd{\gterm{H}}{u_0 \ldots u_k}{\sigma_{-2}}{\pi_2}.
    \end{equation}
    As the original strategies $C_1(\mathbf{u}, \sigma_{-2})$ and $\pi_2$ are proper, $\rho$ must reach $V\times \{1\}$. Let $(v,1)$ be the terminating state in $\gterm{H}$. As $\sigma_{-2}$ is winning for player $1$ in player $2$'s B\"uchi game from $u_k$, and $(v,1)$ is the only state that the path visits infinitely often, then we have $v \in V^-_2(q)$. Thus, player $2$'s deviation does not yield a strictly preferred outcome. Symmetric reasoning holds for player $1$, so the strategy pair is a Nash equilibrium.
\end{proof}

\subsubsection{Necessity of Algorithm~\ref{alg:MainAlg}}
We show that when Algorithm~\ref{alg:MainAlg} returns \textbf{false}, no Nash equilibrium exists in proper strategies in $\gterm{H}$. 

\begin{theorem}
    Suppose Algorithm~\ref{alg:MainAlg} returns \textbf{false}. Then no proper strategy Nash equilibrium exists in $\gterm{H}$. 
\end{theorem}
\begin{proof}
    Let $(\sigma_1,\sigma_2)$ be a pair of proper strategies. By Proposition~\ref{prop:ProperPolicyTerm} there exists some $w\in S\times Q$ such that 
    $\PathsInd{\gterm{H}}{(v_0,0)}{\sigma_1}{\sigma_2}$ reaches $(w,1)$. Because the algorithm returned \textbf{false},
    $\rho = \PathsInd{\gterm{H}}{(v_0,0)}{\sigma_1}{\sigma_2}$ must visit one of the player's winning regions for the B\"uchi games $\mathsf{Buchi}(\bgobj{1}{q})$ and $\mathsf{Buchi}(\bgobj{2}{q})$ that we define for final semi-automaton state $q = \mathsf{Semi}(w)$. Assume without loss of generality that the path visits the winning region of player $2$, and there exists $k$ such that $\rho[k] \notin U_2(q)$. Player $2$ has a potentially non-proper strategy $\pi_2$ to win their B\"uchi objective. We construct a proper strategy using $\pi_2$. 

    \textit{Case 1: } Suppose $V^+_2(q) \times \{1\}$ appears infinitely often in $\rho' = \PathsInd{\gterm{H}}{(v_0,0)}{\sigma_1}{\pi_2}$. In this case, $\rho'$ visits a state $(z,1)$ at some time $T$, such that $z \succ_{\mathcal{E}_2} w$. Define $\pi_2^T$ as 
    \begin{equation}
        \pi_2^T(h) = \begin{cases}
            \pi_2(h) & \text{ if } |h| \leq 2T, \\
            \tau & \text{ if } |h| > 2T.
        \end{cases}
    \end{equation}
    As the path $\rho'$ terminates in $T$ steps, we have 
    $\PathsInd{\gterm{H}}{(v_0,0)}{\sigma_1}{\pi_2^T} = \rho'$.
    However, $\pi_2^T$ is now a proper strategy that leads to an outcome that player $2$ strictly prefers to $w$, and thus there exists a proper deviating policy.

    \textit{Case 2: } 
    Suppose $V^+_2(w) \times \{0\}$ appears infinitely often in $\rho' = 
    \PathsInd{\gterm{H}}{(v_0,0)}{\sigma_1}{\pi_2}$,
    but $V^+_2(w) \times \{1\}$ does not appear.
    In this case, the game does not terminate.
    However, as $\sigma_1$ is proper, we must have that there exists $N$ such that $n\geq N $ implies $\rho'[n] \in V_2$, otherwise, the path must terminate.
    As such, the only states visited infinitely often are in $V_2\times\{0\}$, but as $V^+_2(w) \times \{0\}$ appears infinitely often, it must follow that $\rho'$ visits $(V_2\cap V^+_2(w))\times \{0\}$.
    As such, we can construct a strategy $\hat{\pi}_2$ such that $\PathsInd{\gterm{H}}{(v_0,0)}{\sigma_1}{\hat{\pi}_2}$ visits $V^+_2(w)\times \{1\}$.
    To construct this new strategy, we simply ensure player $2$ terminates the game on reaching $(V_2\cap V^+_2(w))\times \{0\}$. We can then pass to \textit{Case 1}, as we have constructed a strategy $\hat{\pi}_2$ such that the game reaches $V_2^+(w) \times \{1\}$.

    We can conclude that there always exists a proper strategy $\pi_2$, which achieves a strictly preferred outcome for player $2$. Thus, $(\sigma_1,\sigma_2)$ is not a Nash equilibrium strategy pair, and as the strategy pair was arbitrary, we conclude that no such strategy pair exists.
    $\square$
\end{proof}

\section{Privacy proofs}

In this section, we give an extended version of the proofs of Theorem~\ref{thm:privacyGuar} and Lemmas~\ref{lem:prefGenCorr1} and \ref{lem:prefGenCorr2}.

\subsection{Proof of Theorem~\ref{thm:privacyGuar}}

\textit{Outline: }
For a given preorder $E$, we prove this theorem by constructing another preorder $\tilde{E}$ and showing that, when we step through the privacy protocol, the outputs are the same. 
We induct on the state of the algorithm's memory, which comprises the sets $\qmem{j}$, $\dmemmi{j}$, and $\dmempl{j}$ before query $j$. We then prove that, when the memory is the same at each time step, the responses are also the same.

\textit{Proof}
Given a preorder $E \in \frU_{xy}$ and a query sequence $(\qc_1,f_1),\ldots, (\qc_k,f_k)$, we construct a preorder $\tilde{E} \in \frU_{xy}^c$ that produces the same sequence of responses.

Let $n_j$ denote the memory-state of the algorithm before query $j$ when run with preorder $E$.
This memory-state comprises the lists of dummy preferences
\begin{equation}
    \dmempl{j} = (\dppl{1},\ldots,\dppl{j}) \text{ and } \dmemmi{j} = (\dpmi{1},\ldots,\dpmi{j}),
\end{equation}
along with the previous queries
\begin{equation}
    \qmem{j} = (\qc_1,\ldots,\qc_{j-1}).
\end{equation}
That is, $n_j = (\dmempl{j},\dmemmi{j}, \calU_j)$.
We will use $\tilde{n}_j = (\tilde{\Gamma}_{j}^+,\tilde{\Gamma}_j^-,\tilde{\calU}_j)$ to denote the similar values when we run $\resp$ with $\tilde{E}$,
where 
\begin{equation}
    \tilde{\Gamma}_k^+ = (\tilde{\dpc}_1^+,\ldots,\tilde{\dpc}_k^+) \text{ and } \tilde{\Gamma}_k^- = (\tilde{\dpc}_1^-,\ldots,\tilde{\dpc}_k-).
\end{equation}
We note that $n_1 = (\Gamma^+_1,\Gamma^-_1, \calU_1)$ and $\tilde{n}_1 =  (\tilde{\Gamma}^+_1,\tilde{\Gamma}^-_1,\tilde{\calU}_1)$ where
\begin{align}
    & \mathcal{U}_1 = \tilde{\calU}_1 = \emptyset, \\
    & \Gamma^-_1 = \tilde{\Gamma}^-_1 = \{\calE_{\empty}\}, \text{ and}\\
    & \Gamma^+_1 = \tilde{\Gamma}^+_1 = \{\calE_{x,y}\}.
\end{align}

For a fixed secret $\frU_{xy}$, the function $\resp$, in Algorithm~\ref{alg:mergedAlgo}, is a function that, for a given preorder, updates the memory, and produces a response $\rc_j$.
That is, we can define
\begin{equation}
    n_{j+1}, \rc_j = \resp(E, n_j,  (\qc_j,f_j))
\end{equation}
and
\begin{equation}
    \tilde{n}_{j+1}, \tilde{\rc}_j = \resp(\tilde{E}, \tilde{n}_j,  (\qc_j,f_j)),
\end{equation}
where the secret $\frU_{xy}$ is implicit.

We want to show
\begin{equation}
    \rc_j = \tilde{\rc}_j \quad \forall j \in \{1,\ldots, k\}.
\end{equation}

We consider two cases. 
\paragraph{Case 1} 
$\rc_k = \mathsf{STOP}$, but $\rc_j \neq \mathsf{STOP}$ for $j \in \{1,\ldots,k-1\}$. 

Let $n_k = (\dmempl{k}, \dmemmi{k}, \mathcal{U}_k)$, where
$\dmemmi{k} = (\dpmi{1},\ldots,\dpmi{k})$. 

In this case, we set $\tilde{E} = \dpmi{k}$. 
By definition of $\dpmi{k}$ we have 
$
    \dpmi{k} \in \calP^- \setminus (\Gamma_{k}^+ \cup \Gamma_{k}^-),
$
where
\begin{equation*}
    \calP^+,\calP^- = \textsc{\genk}(E, (\qc_1,\ldots, \qc_{k-1}), \frU_{xy}, k, f_{k-1}).
\end{equation*}
Thus we can deduce that
$
    \dpmi{k} \in \frU^{c}_{xy},
$
and 
\begin{multline}
    \label{eq:winEquality}
    (V_r\times \{0\}) \cap U_1^E(\qc_j) = 
     (V_r\times \{0\}) \cap U_1^{\dpmi{k}}(\qc_j)
    \\ \forall j = 1,\ldots, k-1.
\end{multline}
We also have $\dpmi{k} \neq \dpmi{t}$ and $\dpmi{k} \neq \dppl{t}$ for all $t < k$. 
Finally, by the invariance condition, we can deduce that
\begin{multline}
    \label{eq:gkpEqPriv1}
    \textsc{\genk}(E, (\qc_1,\ldots, \qc_j), \frU_{xy}, j+1, r)
    \\ = \textsc{\genk}(\dpmi{k}, (\qc_1,\ldots, \qc_j), \frU_{xy}, j+1, r)
\end{multline}
for all $j = 1,\ldots,k-1$.
For brevity, we will use $g$ to denote the function $\textsc{\genk}$.

For $\tilde{E} = \dpmi{k}$, we first show that the memory states are equal at each iteration, \ie, $n_j = \tilde{n}_j$ for $j = 1$ to $k$. 
By definition, $n_1 = \tilde{n}_1$. 
Suppose $n_j = \tilde{n}_j$ for $j < k$, i.e, 
\begin{itemize}
    \item $\mathcal{U}_j = \tilde{\mathcal{U}}_j$,
    \item $\Gamma_j^+ = \tilde{\Gamma}_j^+$, and 
    \item $\Gamma_j^- = \tilde{\Gamma}_j^-$.
\end{itemize}
We step through $\resp$, and show that the changes to memory are the same for $E$ and $\tilde{E}$.

By assumption $\rc_j \neq \mathsf{STOP}$ for $j < k$, so the if statement on line~\ref{ln:STOPtest} does not trigger with $E$ as the preference. By assumption, $\dpmi{k} \neq \dpmi{t}$ and $\dpmi{k} \neq \dppl{t}$ for all $t < k$.
As $n_j = \tilde{n}_j$, $\tilde{\Gamma}_j^+ = \dmempl{j}$ and $\tilde{\Gamma}_j^- = \dmemmi{j}$, we conclude that $\dpmi{k} \neq \tilde{\dpc}_{j}^+$ and $\dpmi{k} \neq \tilde{\dpc}_{j}^-$, such that line~\ref{ln:STOPtest} does not trigger with $\tilde{E}$ either.

Let $\calP^+_j,\calP^-_j = g(E, (\qc_1,\ldots, \qc_{j}), \frU_{xy}, j+1, f_{j})$, 
and let $\tilde{\calP}^+_j,\tilde{\calP}^-_j = g(\tilde{E}, (\qc_1,\ldots, \qc_{j}), \frU_{xy}, j+1, f_{j})$. 
By \eqref{eq:gkpEqPriv1}, and the fact that $f_j = r$, as $j < k$, we have $\calP^+_j = \tilde{\calP}^+_j$ and $\calP^-_j = \tilde{\calP}^-_j$.

The inductive assumption implies that $\dmempl{j} = \tilde{\Gamma}_j^+$ and $\dmemmi{j} = \tilde{\Gamma}_j^-$, and we argued above that $\calP^+_j = \tilde{\calP}^+_j$ and $\calP^-_j = \tilde{\calP}^-_j$.
Thus, lines \ref{ln:pProcStart}-\ref{ln:pProcEnd} and line \ref{ln:updGamma} in Algorithm~\ref{alg:mergedAlgo} will involves the exact same operations for both $E$ and $\tilde{E}$, which implies $\Gamma^+_{j+1} = \tilde{\Gamma}_{j+1}^+$ and $\Gamma^-_{j+1} = \tilde{\Gamma}_{j+1}^-$. 
As the queries are the same, it also trivially holds that $\calU_{j+1} = \tilde{\calU}_{j+1}$. 

Thus, by the principle of mathematical induction, we can deduce that $n_j = \tilde{n}_j$ for $j = 1,\ldots, k$.

We now argue that $\rc_j = \tilde{\rc}_j$ for all $j = 1,\ldots, k-1$.

As a byproduct of the above reasoning, we have
\begin{equation}
    (\calP_j^+,\calP_j^-) = (\tilde{\calP}_j^+,\tilde{\calP}_j^-)
\end{equation}
for all $j = 1,\ldots, k-1$.
Additionally, the first $k-1$ queries must have the flag as $r$ by assumption, as the mediator will only pass the punishment query as the final query.
Finally, with neither preorder does the protocol stop at line~\ref{ln:STOPtest}, as we argued earlier.

Thus, for both $E$ and $\tilde{E}$, $\resp$ will generate the respective responses $\rc_j$ and $\tilde{\rc}_j$ to the $j^{\text{th}}$ query at line~\ref{ln:wrReturn}.
By \eqref{eq:winEquality} we can then deduce that $\rc_j = \tilde{\rc}_j$

For the case when $j = k$, we have, by assumption that $\rc_k = \mathsf{STOP}$, while, by definition of $\tilde{E}$ as $\dpmi{k}$, and the fact that $n_j = \tilde{n}_j$, we have $\tilde{\rc}_k = \mathsf{STOP}$.

\paragraph{Case 2}
$\rc_k \neq \mathsf{STOP}$. 

In this case, we set $\tilde{\mpre} = \dpmi{k+1}$. 
By definition, we have $ \dpmi{k+1} \in \calP_k^-$ where
\begin{equation*}
    \calP_k^+,\calP_k^- = 
    \textsc{\genk}
    (\mpre, (\qc_1,\ldots,\qc_k), \frU_{xy},k+1,f_k)
\end{equation*}
and
$
    \dpmi{k+1} \notin (\Gamma_k^+ \cup \Gamma_k^-).
$
Thus we can deduce that 
$
    \dpmi{k+1} \in \frU^{c}_{xy},
$
and, 
\begin{multline}
    \label{eq:winEquality2}
    (V_r\times \{0\}) \cap U_1^E(\qc_j) = 
     (V_r\times \{0\}) \cap U_1^{\dpmi{k+1}}(\qc_j)
    \\ \forall j = 1,\ldots, k.
\end{multline}
Also, if $f_k = p$, we have
\begin{equation}
    \label{eq:upEquality}
    \{\qc_k\}^\uparrow_{\dpmi{k+1}} = \{\qc_k\}^\uparrow_{\mpre}.
\end{equation}
Finally, by the invariance condition, we have
\begin{multline}
    \textsc{\genk}(E, (\qc_1,\ldots, \qc_j), \frU_{xy}, j+1, r)
    \\ = \textsc{\genk}(\dpmi{k+1}, (\qc_1,\ldots, \qc_j), \frU_{xy}, j+1, r)
\end{multline}
for all $j = 1,\ldots k-1$.

We now again induct on $j$ to prove that $n_j = \tilde{n}_j$.
By definition of $n_1$ and $\tilde{n}_1$ we have $n_1 = \tilde{n}_1$. 

Suppose $n_j = \tilde{n}_j$. 
By construction $\dpmi{k+1}$ is not contained in the set 
$\dmemmi{j} \cup \dmempl{j}$, and by the inductive assumption $\dmempl{j} = \tilde{\Gamma}_j^+$ and $ \dmemmi{j} = \tilde{\Gamma}_j^-$. Thus, similar to the first case, we can argue that line~\ref{ln:STOPtest} does not trigger.

Let $\calP^+_j,\calP^-_j = g(E, (\qc_1,\ldots, \qc_{j}), \frU_{xy}, j+1, f_{j})$, and let $\tilde{\calP}^+_j,\tilde{\calP}^-_j = g(\tilde{E}, (\qc_1,\ldots, \qc_{j}), \frU_{xy}, j+1, f_{j})$, where we again use $g$ as shorthand for $\textsc{\genk}$. By definition of $\tilde{E} = \dpmi{k+1}$, the fact that $f_j = r$, as $j < k$, and the invariance assumptions on $g$, we have $\calP^+_j = \tilde{\calP}^+_j$ and $\calP^-_j = \tilde{\calP}^-_j$.

Thus, by similar arguments to case $1$, we can deduce that $n_{j+1} = \tilde{n}_{j+1}$, and conclude that $n_{j} = \tilde{n}_j$ for all $j \in \{1,\ldots,k\}$.

We now argue that $\rc_j = \tilde{\rc}_j$ for all $j = 1,\ldots,k-1$.
By assumption, for any $j = 1,\ldots, k-1$, preorder $E$ does not cause $\mathsf{STOP}$ to return at line~\ref{ln:STOPtest}.
The same holds for $\tilde{E}$, as $n_j = \tilde{n}_j$
Furthermore, $\calP^+_j = \tilde{\calP}^+_j$, $\calP_j^- = \tilde{\calP}_j^-$ and $\tilde{n}_j = n_j$ for all $j = 1,\ldots,k-1$, and thus \textsc{\genk} executes identically from line~\ref{ln:pProcStart} to \ref{ln:pProcEnd} for both $E$ and $\tilde{E}$.
Finally, by definition of $\tilde{E}$ as $\dpmi{k+1}$, and the fact that $f_j$ must be $r$ for $j < k$, $\resp$ will generate $\rc_j$ and $\tilde{\rc}_j$ through line~\ref{ln:wrReturn}, and by \eqref{eq:winEquality2}, 
The values $\resp$ generates as a reply will be the same, that is, $\rc_j = \tilde{\rc}_j$.

For the $k^{\text{th}}$ query we remark that, by assumption $\rc_k \neq \mathsf{STOP}$, and so depending on the flag $f_k$, $\rc_k$ will either be the winning region $\lsrP{1}{\qc_k}{E}$ or the B\"uchi game $\mathsf{Buchi}(\bgobjP{1}{\qc}{E})$.

Meanwhile, by definition, $ \dpmi{k+1} \notin \Gamma_k^+ \cup \Gamma_k^- $ and hence $ \dpmi{k+1} \notin \tilde{\Gamma}_k^+ \cup \tilde{\Gamma}_k^-$ and so $\resp$ will not return $\mathsf{STOP}$ at line~\ref{ln:STOPtest}.
Furthermore, again by the invariance condition, $\calP_k^+ = \tilde{\calP}_k^+$ and $\calP_k^- = \tilde{\calP}_k^-$ where 
\begin{align}
    & \calP^+_k,\calP^-_k = g(E, (\qc_1,\ldots, \qc_{k}), \frU_{xy}, k+1, f_{k}),\\
    & \tilde{\calP}_k^+,\tilde{\calP}_k^- = g(\tilde{E}, (\qc_1,\ldots, \qc_{k}), \frU_{xy}, k+1, f_{k}),
\end{align}
and thus $\resp$ will not return $\mathsf{STOP}$ between lines~\ref{ln:pProcStart} and \ref{ln:pProcEnd} on query $k$ for $\tilde{E}$, as these lines execute identically for preorders $E$ and $\tilde{E}$.
Thus, depending on the flag $f_k$, $\tilde{\rc}_k$ will either be the winning region $\lsrP{1}{\qc_k}{\tilde{E}}$ or the B\"uchi game $\mathsf{Buchi}(\bgobjP{1}{\qc}{\tilde{E}})$.
In either case, by \eqref{eq:winEquality2} and \eqref{eq:upEquality},we can deduce that $\rc_k = \tilde{\rc}_k$.

The above proof assumed $E \in \frU$, but symmetric reasoning holds for the case where $E \in \frU^c$.
$\square$

\subsection{Proofs of Lemmas~\ref{lem:prefGenCorr1} and \ref{lem:prefGenCorr2}}

\begin{lemma}[Restatement of Lemma~\ref{lem:prefGenCorr1}]
    $\textsc{CheckConstraints}(Q, \mathcal{C}, \mathcal{E}_{init})$ returns the smallest preorder satisfying constraints $\mathcal{C}$ that contains $\mathcal{E}_{init}$
\end{lemma}
\begin{proof}
We first note that $\mathcal{E}$ is always a subset of any preference relation that satisfies $\mathcal{C}$ and contains $\calE_{init}$.
Indeed, we can deduce this fact by inducting on the construction of $\calE$ in \textsc{CheckConstraints}.

We next note that \textsc{CheckConstraints} terminates, that is, always returns $\mathsf{Pass}$ or $\mathsf{Fail}$.
Indeed, this claim can only fail if the while loop runs indefinitely. 
However, if the algorithm does not terminate for a given loop instance, at least one edge must be added to $\mathcal{E}$.
As such, $\mathcal{E}$ will eventually be the complete graph, and for the complete graph, the while loop must terminate via \textbf{return}.

Finally, we argue that if the algorithm returns $\mathcal{E}$ then $\mathcal{E}$ is a preorder satisfying the constraints, and if the algorithm returns $\mathsf{Fail}$, no such preorder exists that also contains $\calE_{init}$.
Indeed, if the algorithm returns $\mathcal{E}$ after a given loop iteration, $\mathcal{E}$
is a preorder, and satisfies all constraints.
Alternatively, if the algorithm returns $\mathsf{Fail}$ there exists a node pair $u,v$ such that $(u,v) \in \mathcal{E} $ and $ (v,u)\in\mathcal{E} $, and $(u,v \succ) \in \mathcal{C}$. 
However, any relation satisfying $\mathcal{C}$ and containing $\calE_{init}$ must contain $\mathcal{E}$, as we argued at the beginning of this proof, and thus can not satisfy $u \succ v$ . 
We can then deduce that no such preorder exists.
$\square$
\end{proof}

We additionally note that \textsc{CheckConstraints} runs in polynomial time.
\begin{lemma}
    \textsc{CheckConstraints} terminates in polynomial time in the number $|Q|$ of semi-automaton states.
\end{lemma}
\begin{proof}
    This fact follows from the fact that, at each stage where the algorithm does not terminate, the size of $\mathcal{E}$ increases by $1$. However, this can only happen $|Q|^2$ times.
    The individual steps in each loop also run in polynomial time.
    $\square$
\end{proof}

\begin{lemma}[Restatement of Lemma~\ref{lem:prefGenCorr2}]
    Let
    $ \mathcal{P} = \textsc{\genfc}(Q,\mathcal{C},K)$.
    Either $|\mathcal{P}| \geq K$ or $\mathcal{P}$ contains all preferences that satisfy the constraints in $\mathcal{C}$.
\end{lemma}

\begin{proof}
    All of the preorders that \textsc{\genfc} adds to $\mathcal{P}$ satisfy the constraints, by the properties of \textsc{CheckConstraints}.
    Now assume that the initial preorder is not $\mathsf{Fail}$, and consider the while loop.
    For each iteration of the while loop, either \npvar \  is \textbf{false} at the end of the loop iteration, or the size of $\mathcal{P}$ increases. 
    Now consider the case where \npvar \  is \textbf{false} at the end of the loop, that is, $\mathcal{P}$ is constant throughout the loop.
    We claim that $\mathcal{P}$ contains all preorders that satisfy the constraints. 
    Indeed, suppose by contradiction that there exists a preorder $E$ on $Q$ which is not in $\mathcal{P}$ but satisfies the constraints.
    Let $E' = \textsc{CheckConstraints}(Q,\mathcal{C},\emptyset)$.
    $E'$ is a preorder, $E' \subset E$ by Lemma~\ref{lem:prefGenCorr1}, and $E' \in \calP$.
    Now take a preorder $E''$ in $\mathcal{P}$ such that
    \begin{enumerate}
        \item $E'' \subset E$
        \item There does not exist a preorder $B \in \mathcal{P}$ such that $B \supsetneq E''$ and $B \subset E$.
    \end{enumerate}
    Such a preorder exists as we can construct it by induction.
    Now let $e^*$ be an edge such that $e^* \in E \setminus E''$, which exists as $E'' \subsetneq E$.
    At some point in the course of the while loop iteration we will run \textsc{CheckConstraints}($Q,\mathcal{C},E'' \cup \{e^*\}$).
    We know that this call will not return $\mathsf{Fail}$ as there exists a preorder $E$ that satisfies the constraints $\mathcal{C}$ and contains $E'' \cup \{e^*\}$.
    Let $B' = \textsc{CheckConstraints}(Q,\mathcal{C},E'' \cup \{e^*\})$.
    We must have $B' \notin \mathcal{P}$, otherwise we would violate condition 2. above on $E''$, \ie, that $E''$ is not a subset of another preorder in $\mathcal{P}$. 
    However, this new preorder $B'$ would have caused \npvar \  to be set to \textbf{true}, and thus we have arrived at a contradiction.

    If \npvar \  is ever \textbf{false} at the end of the while loop, the while loop will exit, the algorithm will return $\mathcal{P}$, and $\mathcal{P}$ will contain all preorders that satisfy the constraints $\mathcal{C}$.
    In the other case, \npvar \  is never \textbf{false}, and on each iteration, $\mathcal{P}$ increases in size by one. 
    Thus, we will eventually have $|\mathcal{P}| \geq K$. 
\end{proof}

\begin{remark}
    \textsc{\genfc} makes at most $k^2 |Q|^2$ calls to \textsc{CheckConstraints}.
\end{remark}

\subsection{Implementing \textsc{\genk}}

In this section of the appendix, we give extra details on the implementation and properties of \textsc{\genk}.

\begin{proposition}
    Algorithm~\ref{alg:genKImp} satisfies Assumption~\ref{ass:GenKAss}.
\end{proposition}
\begin{proof}
    By correctness of \textsc{\genfc}, any preorder $E'$ in $\calP^+$ satisfies the constraint $(x,y,\succ)$, \ie,  $x\succ_E y$, which implies $\calP^+ \subset \frU_{xy}$. 
    The same argument implies $\calP^- \subset \frU_{xy}^c$.
    Trivially, $|\calP^+| \leq K$ and $|\calP^-| \leq K$.

    Additionally, by correctness of \textsc{\genfc}, any preorder $E' \in \calP^+ \cup \calP^-$ satisfies
    \begin{equation}
        v \succ_{E'} \qc_i \quad \forall v \in \upc{\qc_i}{E},  i = 1,\ldots, k,
    \end{equation}
    and
    \begin{equation}
        v \not\succ_{E'} \qc_i \quad \forall v \notin \upc{\qc_i}{E}, i = 1,\ldots, k.
    \end{equation}
    Thus we deduce that, for all $E' \in \calP^+ \cup \calP^-$,
    \begin{equation}
        \upc{\qc_i}{E'} = \upc{\qc_i}{E} \quad \forall i = 1,\ldots, k,
    \end{equation}
    and thus $\bgobjP{1}{\qc_i}{E} = \bgobjP{1}{\qc_i}{E'}$ for all $i = 1,\ldots, k$. This fact implies $\lsrP{1}{\qc_i}{E'} = \lsrP{1}{\qc_i}{E}$ for all $i = 1,\ldots,k$.
    Thus, we have proven that  Algorithm~\ref{alg:genKImp} satisfies the first three conditions.

    \textsc{\genk} only depends on its arguments through the upper closures, \ie, there exists some function $\gamma$ such that, when we implement \textsc{\genk} through Algorithm~\ref{alg:genKImp}
    \begin{multline}
        \textsc{\genk}(E,(\qc_1,\ldots,\qc_j), \frU_{xy}, K,f) \\
        = 
        \gamma(\upc{\qc_1}{E},\ldots,\upc{\qc_j}{E}, \frU_{xy}, K),
    \end{multline}
    for any natural integer $K$.
    For any $E' \in \calP^+ \cup \calP^-$, we have
    \begin{equation}
        \upc{\qc_i}{E'} = \upc{\qc_i}{E} \quad \forall i = 1,\ldots, k,
    \end{equation}
    and thus
    \begin{multline}
        \textsc{\genk}(E,(\qc_1,\ldots,\qc_j), \frU_{xy}, K,f) \\
        = 
        \gamma(\upc{\qc_1}{E},\ldots,\upc{\qc_j}{E}, \frU_{xy},K) \\
        =
        \gamma(\upc{\qc_1}{E'},\ldots,\upc{\qc_j}{E'}, \frU_{xy}, K) \\
        =
        \textsc{\genk}(E',(\qc_1,\ldots,\qc_j), \frU_{xy}, K,f)
    \end{multline}
    for any natural integer $K$.
    $\square$
\end{proof}

\section{Experimental details}

In this section, we elaborate on the experiments listed in Section~\ref{sec:privPerm}.

\subsection{Representations of the Shared-Delivery Example}
\label{sec:prodComp}

We first discuss the representation of the product game for the shared-delivery example.

For the shared delivery example, the preference automaton portion of the state is always equal to the state of the original game. Indeed, the preference automaton simply tracks the last state visited. As such, suppose we follow a path $\rho = (l_0, p_0), \ldots, (l_k,p_k)$ in the shared delivery example, where $l_j$ and $p_j$ are a sequence of locations and ownerships of the package. Then it follows that $\delta(q_0, \rho) = l_k$. That is, the semi-automaton portion of the product game state duplicates the location in the game state.

Because of this duplication of the game state in the preference automaton state, the product game and the original game are the same. That is, we can, by an abuse of notation, set $H_\tau = G_\tau$ and define preferences on the states of the product game directly from the preference relations, i.e.,
\begin{equation}
    (l,p) \succeq_{\calE_i} (l',p') \iff l \succeq_{E_i} l'.
\end{equation}

We use this representation of the product game in the following sections to characterize the winning regions of the grid game we study in Section~\ref{sec:privPerm}.

\subsection{Methods for sampling preferences and query sequences}

\subsubsection{Generating random preferences}

We generate random preferences by randomly sampling directed trees according to Algorithm~\ref{alg:GenerateRandomTrees}.
The function \textsc{GenerateRandomUndirectedTree} randomly samples a Pr\"ufer sequence on $|Q|$ nodes, and converts it to a random tree \cite{gottlieb2001prufer}.
We transform the random undirected tree into a directed tree by selecting a root node and conducting a breadth-first search in the tree from that root node.
For each edge that we add to the directed tree, we add it to the preference relation in reverse order.
For example, suppose we have an undirected tree $T$ on $\{1,2,3,4,5\}$ with edges
\begin{equation}
    \{\{1,2\},\{1,3\},\{1,4\},\{4,5\}\}.
\end{equation}
When we root this tree at node $1$, and convert it into a preorder, we obtain the preorder
\begin{multline}
    \calE = \{(2,1),(3,1),(4,1),(5,4),(5,1), \\
    (1,1),(2,2),(3,3),(4,4),(5,5)
    \}.
\end{multline}
That is, node $1$ becomes the least preferred node.

\begin{algorithm}
\caption{Generating random preorders from trees on a set $Q$.}
\label{alg:GenerateRandomTrees}
\begin{algorithmic}[1]
\STATE \textsc{GenerateRandomTree}($Q$)
        \STATE $\calE \gets \emptyset$
        \STATE $\hat{\mathcal{E}} \gets \textsc{GenerateRandomUndirectedTree}($Q$)$ 
        \STATE $q_{\text{root}} \gets \textsc{RandomSample}(Q)$
        \STATE Convert undirected tree $\hat{\mathcal{E}}$ to a directed tree  rooted at $q_{\text{root}}$, and add edges to $\calE$ in reverse order
        \STATE $\mathcal{E} \gets \textsc{TransitiveClosure}(\mathcal{E})$
        \STATE $\calE \gets \calE \cup \{(q,q)| q \in Q\}$
    \RETURN $\calE$
\end{algorithmic}
\end{algorithm}

We sample random trees as they represent a class of preference relations that can still encode incomparability, while being easy to uniformly sample from.

\subsubsection{Sampling random query sequences}

We sample random query sequences by randomly shuffling the list of nodes $Q$ while ensuring the nodes $x$ and $y$ that define the secret $\frU_{xy}$ are at the end of this list, and we detail this procedure in Algorithm~\ref{alg:GenerateRandomSeqs}.
\begin{algorithm}
\caption{Generating random query sequences.}
\label{alg:GenerateRandomSeqs}
\begin{algorithmic}[1]
\STATE \textsc{GenerateRandomQueries}($Q, x , y$)
    \STATE $\hat{Q} \gets Q \setminus \{x,y\}$
    \STATE $\mathcal{U} \gets \textsc{RandomShuffle}(\hat{Q})$
    \STATE Append $x$ and $y$ to $\mathcal{U}$
    \RETURN $\mathcal{U}$
\end{algorithmic}
\end{algorithm}

This method of generating queries helps ensure that the protocol does not terminate prematurely. We note that, in practice, we may consider the secret, i.e., $\frU_{xy}$, to be known to the mediator. 
Note that this assumption does not imply the mediator can determine whether a player's preorder lies in $\frU_{xy}$ or $\frU^c_{xy}$.

\subsection{A game where we can generate winning regions}

In our experiments we use a specific instance of the shared-delivery example to investigate whether we can make more queries without violating privacy. 

We use an example that represents a grid in which player $1$ can only traverse rows while player $2$ can only traverse columns.
In particular, we define
\begin{equation}
    \mathcal{T} = \{(i,j)| i \in \{1,\ldots,k\}, j \in \{1,\ldots,k\}\}
\end{equation}
for some $k \in \mathbb{N}$.
We define $\mathcal{F}_i$ as follows
\begin{align}
    & \mathcal{F}_1 = \{ ((i,j),(i',j')) | i = i' \wedge |j - j'| = 1\} \cup \calF_{sl}, \\
    & \mathcal{F}_2 = \{ ((i,j),(i',j')) | |i - i'| = 1 \wedge j = j' \}\cup \calF_{sl},
\end{align}
where $\calF_{sl}$ is the set of self-loops on $\calT$.
That is, player $1$ can only traverse rows, while player $2$ can only traverse columns.
The state of the game is a triple $(i,j,p)$ comprising a row $i$, column $j$, and a flag $p \in \{p_1,p_2\}$ indicating the player in control of the package.
In this game $Q = \calT$, that is, players have preferences over the location of the package, and $V_r = \calT$, \ie, all nodes are reachable. 

The set of the states of the terminating product game, when we use the representation from Section~\ref{sec:prodComp}, is $\calT \times \{p_1,p_2\} \times \{0,1\}$ where the flag in $\{0,1\}$ indicates termination.

\subsubsection{Characterizing the winning regions}

In Proposition~\ref{prop:gridGameWinningCharac}, we characterize the set of preorders which share winning regions with $E$.
For a semi-automaton state $\qc$, a preorder $E$ on $Q$ will induce a winning region $W = \calT \times \{p_1,p_2\} \times \{0,1\} \setminus \lsrP{1}{\qc}{E}$.
In Proposition~\ref{prop:gridGameWinningCharac}, we characterize the sets $Q^+ \subseteq Q$, such that a certain B\"uchi objective induces the same winning regions as $E$.
The conditions on these sets $Q^+$ represent conditions on the upper closures of preorders so that they share winning regions with $E$.
That is, given $W$, $\upc{\qc}{E}$ satisfies the conditions in Proposition~\ref{prop:gridGameWinningCharac} if and only if $$W \cap (\calT \times \{0\}) = (\calT \times \{p_1,p_2\} \times \{0\}) \setminus \lsrP{1}{\qc}{E}.$$
Note that Proposition~\ref{prop:gridGameWinningCharac} only depends on $W \cap (\calT \times \{0\})$.

\begin{proposition}
    \label{prop:gridGameWinningCharac}
    Let $W$ be the winning region for player $1$ of the B\"uchi game on the terminating product game arena, where the objective is $Q^+ \times \{p_1,p_2\} \times \{0,1\}$.
    \begin{align}
        & (i,j,p_1,0) \in W \iff \exists j' : (i,j') \in Q^+
        \\& (i,j,p_2,0) \in W \iff \forall i' : (i',j) \in Q^+
    \end{align}
\end{proposition}
\begin{proof}
    \textit{(Statement 1 $\implies$)}
    We go by contraposition.
    We assume 
     $(i,j') \notin Q^+$ for all $j' \in \{1,\ldots,k\}$, and want to prove that $(i,j,p_1,0)$ is not in the winning region.
     Player $1$ can only move in row $i$ while staying in their nodes, and thus can not win the game without passing control to player $2$. 
     However, if player $1$ passes control to player $2$, player $2$ can terminate the game in the corresponding node and force player $1$ to lose.

     \textit{(Statement 1 $\impliedby$)}
     If there exists $j'$ such that $(i,j') \in Q^+$, then player $1$ can simply implement a strategy that induces the path which starts at $(i,j,p_1,0)$, goes to $(i,j',p_1,0)$, and then terminates at $(i,j',p_1,1)$, thus winning the game.

    \textit{(Statement 2 $\implies$)}
    We go by contraposition.
    We assume $(i',j) \notin Q^+$ for some $i'$, and we aim to prove that $(i,j,p_2,1) \notin W$.
    Indeed, from $(i,j,p_2,0)$ player $2$ can simply force the game state to $(i',j,p_2,1)$, and force player $1$ to lose in their B\"uchi game.

    \textit{(Statement 2 $\impliedby$)}
    If $(i',j) \in Q^+$ for all $i' $, then if player $2$ never transfers control to player $1$, then player $1$ wins, as the game is in player $1$'s winning region for all time steps.
    If player $2$ does transfer control, then they transition the game to $(i',j,p_1,0)$ and by assumption $(i',j) \in Q^+$ so player $1$ can then terminate the game, and win. 
    $\square$

\end{proof}

We can use this winning region characterization to search for preorders that induce the same winning regions as some reference preorder $E$.

\subsubsection{Generating preorders that share winning regions}

In this section, we give an implementation of \textsc{\genk} that uses the above characterization of winning regions.

Algorithm~\ref{alg:GeneratingThresholdsFromWinningRegion} searches the powerset of $Q$ for up to $l$ sets $Q^+$ that induce a reference winning region $W$.

Algorithm~\ref{alg:GeneratingPreferencesWithWinningRegions} takes a preorder $E$, and a query list $(\qc_j)_{j=1}^k$ and generates preorders $E'$ such that 
\begin{equation}
    (\calT \times \{0\}) \cap \lsrP{1}{\qc_j}{E} = (\calT \times \{0\}) \cap \lsrP{1}{\qc_j}{E'} 
\end{equation}
for all $j$ in $1,\ldots, k$.
For each winning region $W_j$, restricted to $\calT \times \{0\}$, Algorithm~\ref {alg:GeneratingPreferencesWithWinningRegions} generates a set of possible constraints on upper closures using Algorithm~\ref{alg:GeneratingThresholdsFromWinningRegion}.
Algorithm~\ref{alg:GeneratingPreferencesWithWinningRegions} then searches over combinations of these constraints and applies \textsc{\genfc} to generate preorders.
We implement the search over the product set of threshold constraints lexicographically, but we remark that different search orderings may yield different results.

While the algorithms below do not run in polynomial time, we can trade off computation time and the likelihood of finding preorders by tuning the value of the constant $l$ in Algorithm~\ref{alg:GeneratingPreferencesWithWinningRegions}.

\begin{algorithm}
\caption{Generating upper closure constraints from winning regions.}
\label{alg:GeneratingThresholdsFromWinningRegion}
\begin{algorithmic}[1]
\STATE \textsc{GenerateThresholds}($W$, $l$)
    \STATE $\mathcal{Q} \gets \emptyset$
    \FOR{$Q^+ \in 2^Q$}
        \IF{$Q^+$ satisfies constraints in Proposition~\ref{prop:gridGameWinningCharac}}
            \STATE $\mathcal{Q} \gets \mathcal{Q} \cup \{Q^+\}$
        \ENDIF
        \IF{$|\mathcal{Q}| \geq l$}
            \STATE \textbf{break}
        \ENDIF
    \ENDFOR
    \RETURN $\mathcal{Q}$
\end{algorithmic}
\end{algorithm}

\begin{algorithm}
\caption{Implementation of \textsc{\genk} for \emph{winning-region-privacy}.}
\label{alg:GeneratingPreferencesWithWinningRegions}
\begin{algorithmic}[1]
\STATE \textsc{\genk}$_l$($E,(\qc_1,\ldots,\qc_k),\frU_{xy},K,f$)
    \FOR{$j = 1,\ldots, k-1$}
        \STATE $W_j \gets (\calT \times \{0\}) \setminus \lsrP{1}{\qc_j}{E}$
        \STATE $\mathcal{Q}_j \gets \textsc{GenerateThresholds}(W_j,l)$
    \ENDFOR
    \IF{$f = p$}
        \STATE $\mathcal{Q}_k \gets \{\{\qc_k\}_E^\uparrow\}$
    \ELSE
        \STATE $W_k \gets (\calT \times \{0\}) \setminus \lsrP{1}{\qc_k}{E}$
        \STATE $\mathcal{Q}_k \gets \textsc{GenerateThresholds}(W_k ,l)$ 
    \ENDIF
    \STATE $\mathcal{P}^+ \gets \emptyset$
    \STATE $\mathcal{P}^- \gets \emptyset$
    \STATE $c \gets 0$
    \FOR{$(Q_1,\ldots,Q_k) \in \mathcal{Q}_1 \times \ldots, \times \mathcal{Q}_k$}
        \STATE $c \gets c + 1$
        \STATE $\mathcal{C} \gets \emptyset$
        \FOR{$j = 1\ldots,k$}
            \STATE $\mathcal{C} \gets \mathcal{C} \cup \{(v,\qc_j,\succ)| v \in Q_j\}$
            \STATE $\mathcal{C} \gets \mathcal{C} \cup \{(v,\qc_j,\not\succ)| v \notin Q_j\}$
        \ENDFOR
        \STATE $\mathcal{C}^+ \gets \mathcal{C} \cup \{(x,y,\succ)\}$
        \STATE $\mathcal{C}^- \gets \mathcal{C} \cup \{(x,y,\not\succ)\}$
        \STATE $\mathcal{P}^+ \gets \calP^+ \cup \textsc{\genfc}(Q,\calC^+,K)$
        \STATE $\mathcal{P}^- \gets \calP^- \cup \textsc{\genfc}(Q,\calC^-,K)$
        \IF{$c \geq l$ or ($|\mathcal{P}^+| > K$ and $|\mathcal{P}^-| > K$)}
            \STATE \textbf{break}
        \ENDIF
    \ENDFOR
    \STATE Truncate $\calP^+,\calP^-$ to $K$ values
    \RETURN $\mathcal{P}^+,\calP^-$
\end{algorithmic}
\end{algorithm}

For completeness, we validate that Algorithm~\ref{alg:GeneratingPreferencesWithWinningRegions} satisfies Assumption~\ref{ass:GenKAss}, but we remark that this fact is already clear from the statement of Algorithm~\ref{alg:GeneratingPreferencesWithWinningRegions}.
\begin{proposition}
    Algorithm~\ref{alg:GeneratingPreferencesWithWinningRegions} satisfies Assumption~\ref{ass:GenKAss}.
\end{proposition}
\begin{proof}
    Trivially, $\calP^+$ and $\calP^-$ satisfy condition 1) in Assumption~\ref{ass:GenKAss} as the constraint sets from which we generate $\calP^+$ and $\calP^-$ contain $(x,y,\succ)$ and $(x,y,\not\succ)$ respectively.
    
    For condition 2), first let $E'$ be an element of $\calP^+ \cup \calP^-$ and note that, for all $j = 1,\ldots,k$,
    \begin{equation}
        \upc{\qc_j}{E'} \in \calQ_j.
    \end{equation}
    By Proposition~\ref{prop:gridGameWinningCharac} and Algorithm~\ref{alg:GeneratingThresholdsFromWinningRegion}, it then follows that
    \begin{equation}
        \lsrP{1}{\qc_j}{E'} \cap (V_r \times \{0\})
        = \lsrP{1}{\qc_j}{E} \cap (V_r \times \{0\}),
    \end{equation}
    where $V_r = \calT$.
    Note that the above reasoning incorporates the case where $\calQ_k = \{\upc{\qc_k}{E}\}$.

    Condition 3) also holds, as we define $\calQ_k$ conditional on $f$, and if $f = p$, the constraints on $E'$ imply that $\upc{\qc_k}{E} = \upc{\qc_k}{E'}$.

    Finally, regarding the invariance condition 4), we consider two cases.
    
    If $f = r$, Algorithm~\ref{alg:GeneratingPreferencesWithWinningRegions} is only a function of preorder $E$ thorough the losing region intersections $(V_r \times \{0\} ) \cap (\lsrP{1}{\qc_j}{E})_{j=1}^k$, and any preorder $E'$ that Algorithm~\ref{alg:GeneratingPreferencesWithWinningRegions} returns through $\calP^+$ and $\calP^-$ satisfies $$\lsrP{1}{\qc_j}{E'} \cap (V_r \times \{0\} )
        = \lsrP{1}{\qc_j}{E} \cap (V_r \times \{0\})$$
        for all $j$ in $1,\ldots,k$.
    Thus, for all $j = 1,\ldots,k$ and natural integer $K$,
    \begin{multline}
        \textsc{\genk}(E,(q_1,\ldots,q_j),\frU_{xy},K,f_j) \\ = \textsc{\genk}(E',(q_1,\ldots,q_j),\frU_{xy},K,f_j)
    \end{multline}
    where $f_j = r$ for all $j = 1,\ldots,k$.

    If $f = p$, Algorithm~\ref{alg:GeneratingPreferencesWithWinningRegions} is only a function of preorder $E$ through the losing region intersections $ (V_r \times \{0\} ) \cap(\lsrP{1}{\qc_j}{E})_{j=1}^{k-1}$ and the upper closure $\upc{\qc_k}{E}$.
    Any preference that Algorithm~\ref{alg:GeneratingPreferencesWithWinningRegions} returns through $\calP^+$ and $\calP^-$ satisfies $$\lsrP{1}{\qc_j}{E'} \cap (V_r \times \{0\} )
        = \lsrP{1}{\qc_j}{E} \cap (V_r \times \{0\}) \quad $$
    for all $j$ in $1,\ldots, k -1$,
    and $\upc{\qc_k}{E} = \upc{\qc_k}{E'}$.
    Thus, for all $j = 1,\ldots, k-1$ and natural integer $K$,
    \begin{multline}
        \textsc{\genk}(E,(q_1,\ldots,q_j),\frU_{xy},K,r) \\= \textsc{\genk}(E',(q_1,\ldots,q_j),\frU_{xy},K,r),
    \end{multline}
    and
    \begin{multline}
        \textsc{\genk}(E,(q_1,\ldots,q_k),\frU_{xy},K,p) \\= \textsc{\genk}(E',(q_1,\ldots,q_k),\frU_{xy},K,p).
    \end{multline}
    $\square$
\end{proof}

\end{document}